\documentclass[letterpaper, 11 pt, onecolumn]{IEEEtran}
\usepackage{amsmath,amsfonts,amssymb}
\usepackage{graphicx,epsfig,psfrag,subfigure,remark}
\usepackage{xfrac,bm,savesym,xcolor}

\usepackage{algorithm}
\usepackage[noend]{algpseudocode}

\usepackage[all]{xy}
\usepackage{varioref}
\usepackage{wrapfig}
\usepackage{threeparttable}
\usepackage{dcolumn}
\newcolumntype{d}{D{.}{.}{-1}}
\usepackage{nomencl}
\makeglossary
\usepackage{subfigure}
\usepackage{comment,cite}
\IEEEoverridecommandlockouts 

\newcommand{\x}{{x}}

\newcommand{\defn}{:=}

\newcommand{\e}{e}

\newcommand{\cA}{\mathcal{A}}
\newcommand{\cB}{\mathcal{B}}
\newcommand{\cC}{\mathcal{C}}
\newcommand{\cD}{\mathcal{D}}
\newcommand{\cE}{\mathcal{E}}
\newcommand{\cI}{\mathcal{I}}

\newcommand{\cM}{\mathcal{M}}
\newcommand{\cN}{\mathcal{N}}
\newcommand{\cP}{\mathcal{P}}

\newcommand{\cO}{\mathcal{O}}

\newcommand{\cS}{\mathcal{S}}

\newcommand{\mB}{\mathfrak{B}}
\newcommand{\mZ}{\bm{\mathbb{Z}}}

\newcommand{\mx}{\mathsf{x}}

\newcommand{\mz}{\mathsf{z}}

\newcommand{\mP}{\mathfrak{P}}
\newcommand{\mfp}{\mathfrak{p}}
\newcommand{\mff}{\mathfrak{f}}
\newcommand{\mfg}{\mathfrak{g}}

\newcommand{\mT}{T}

\newcommand{\cV}{\mathcal{V}}

\newcommand{\intr}{\ensuremath{\operatorname{int}}}

\newtheorem{proposition}{Proposition}

\newtheorem{definition}{Definition}
\newtheorem{problem}{Problem}

\newtheorem{assumption}{Assumption}

\renewcommand{\t}{^{\mbox{\scriptsize \mT}}}
\newcommand{\bd}{\mathrm{bd}}

\makeatletter
\def\BState{\State\hskip-\ALG@thistlm}
\makeatother

\renewcommand{\baselinestretch}{1.00}
\newremark{remark}{Remark}

\title{Workspace Partitioning and Topology Discovery Algorithms for Heterogeneous Multi-Agent Networks} 

\author{Efstathios Bakolas \thanks{E. Bakolas
is an Associate Professor in the Department of Aerospace Engineering
and Engineering Mechanics, The University of Texas at Austin,
Austin, Texas 78712-1221, USA, Email: bakolas@austin.utexas.edu.
This work was supported in part by the National Science Foundation
(award no. CMMI-1753687).}}

\begin{document}

\maketitle

\begin{abstract} In this paper, we consider a class of workspace
partitioning problems that arise in the context of area coverage and
spatial load balancing for spatially distributed heterogeneous
multi-agent networks. It is assumed that each agent has certain
directions of motion or directions for sensing and exploration that
are more preferable than others. These preferences are measured by
means of convex and anisotropic (direction-dependent) quadratic
proximity metrics which are, in general, different for each agent.
These proximity metrics induce Voronoi-like partitions of the
network's workspace 
that are comprised of cells which may not always be convex (or even connected) sets but are necessarily contained in ellipsoids that are known to their
corresponding agents. The main contributions of this work are 1) a
distributed algorithm for the computation of a Voronoi-like
partition of the workspace of a heterogeneous multi-agent network
and 2) a systematic process to discover the network topology induced
by the latter Voronoi-like partition. 
Numerical simulations that illustrate the efficacy of the proposed
algorithms are also presented.
\end{abstract}

\section{Introduction}\label{s:intro}
Area coverage and spatial load balancing correspond to two
fundamental classes of problems for spatially distributed
multi-agent networks. Such problems are typically addressed by means
of distributed control algorithms that rely on the use of Voronoi or
Voronoi-like (also known as generalized Voronoi) partitions of the
workspace of the multi-agent network. For the distributed
implementation of these algorithms, each agent has to rely on
information encoded in its own cell from the spatial partition and
perhaps the cells of its neighbors. However, unless the Voronoi-like
partitions are computed by means of distributed partitioning
algorithms, the induced control algorithms are not truly
distributed. Therefore, the development of distributed partitioning
algorithms constitutes an integral component of any
Voronoi-distributed control architecture for a multi-agent network.
A partitioning algorithm can be characterized as distributed when
each agent can compute its own cell independently from its teammates
without utilizing a global reference frame while relying on exchange
of information with only a subset of them (e.g., those that lie
within its communication or sensing range). Ideally, an agent can
compute its own cell if it can exchange information with the agents
that correspond to its neighbors in the topology of the Voronoi-like
partition; these neighboring relations, however, are unknown before
the computation of the Voronoi-partition itself. We will refer to
the problem of characterizing the set of neighbors (or more
realistically, a superset of the latter set) in the topology induced
by the Voronoi-like partition as the ``network topology discovery
problem.'' 

In this work, we propose distributed algorithms that 1) compute
Voronoi-like partitions of the workspace of spatially distributed
heterogeneous multi-agent networks and 2) discover the network
topology induced by the latter partitions. In our approach, the
agents are allowed to have different preferences (hence the
qualifier ``heterogeneous'') which are measured in terms of relevant
proximity (generalized) metrics such as the sensing cost that an
agent will incur to obtain measurements from an arbitrary point in
its spatial domain or the transition cost (e.g., fuel or battery /
energy consumption) that will have to incur to reach it. In our
approach, we assume that the proximity metric associated with an
agent can be expressed as the sum of a convex quadratic form
associated with a positive definite matrix, which we refer to as
\textit{distance operator}~\cite{p:Labelle03}, and a constant term,
which we refer to as \textit{additive gain}. The distance operators
are not necessarily the same for all the agents given that their
workspace may exhibit anisotropic features (e.g., certain directions
of motion or exploration/sensing are more preferable than others).
Some characteristic examples of anisotropic workspaces are oceanic
environments, atmospheric domains and hilly terrains in which
anisotropic features are induced by ocean currents, winds and
elevation variance, respectively. Typically, such anisotropic
features are spatially varying and thus it is natural to associate
each agent with a different distance operator. We will refer to the
Voronoi-like partition of the workspace of a multi-agent network
whose agents utilize proximity metrics with different distance
operators as the Heterogeneous Quadratic Voronoi Partition (HQVP).
In general, the cells that comprise the HQVP may not be convex, or
even connected, sets. Consequently, the computation of HQVP and the
discovery of the induced network topology is not a straightforward
task in sharp contrast with standard Voronoi partitions or other
classes of well studied Voronoi-like partitions (e.g., power
diagrams).

\textit{Literature review:} Area coverage and spatial load balancing
problems for multi-agent networks have received significant
attention in the relevant literature. A well received approach which
leverages the so-called Lloyd's algorithm~\cite{p:Lloyd82} together
with sequences of standard Voronoi partitions can be found
in~\cite{p:cortez}. Several extensions of \cite{p:cortez} have
appeared in the relevant literature (see, for
instance,~\cite{p:cortes05,p:martinez06,p:slotine2009,p:cortes2010,p:Breiten10,p:pavoneTAC2011,p:Schwager2011,fb-rc-pf:08u,rp-pf-rb:13i,p:tzes13,p:bhatta2014,p:tvlloyd15}).
The aforementioned papers deal with multi-agent networks that are
homogeneous in the sense that all of their agents employ the same
proximity metric modulo, perhaps, a different constant term
(additive gain). In this work, a multi-agent network will not be
classified as heterogeneous unless at least two of its agents have
different distance operators and regardless if their additive gains
are the same or not. Coverage problems for heterogeneous networks
with different distance operators
are considered in \cite{p:anisosens2008} based on, however,
centralized techniques. Finally, the problem of discovering the
neighbors of an agent in the topology induced by a standard Voronoi
partition has been studied in \cite{p:hadji,p:lynchvoro2017}. The
applicability of the methods proposed in these references is limited
to standard Voronoi partitions and cannot be extended to the class
of spatial partitions considered in this paper.

In our previous work, we have addressed workspace partitioning
problems for area coverage by homogeneous multi-agent networks based
on proximity (generalized) metrics corresponding to the optimal
cost-to-go functions of relevant optimal control
problems~\cite{p:bt_autom10,p:baktsiaut13,p:bak2018}. In the special
case of linear quadratic optimal control problems, the latter
metrics correspond to convex quadratic functions whose associated
distance operators are, however, the same for all them. Under this
strong assumption, the induced Voronoi-like partitions admit a
special structure that renders them amenable to computation by means
of simple decentralized or distributed
algorithms~\cite{p:bakolas2013b,p:BAKOLAS2014,p:BAKOLAS2016}. The
problem of inferring the neighbors of an agent in the topology
induced by these class of spatial partitions is studied in
\cite{p:BAKOLAS2016,p:bak2018}. 

\textit{Statement of contributions:} The main contribution of this
work is two-fold. First, we show that under some mild technical
assumptions, each cell of the proposed Voronoi-like partition is
necessarily contained inside an ellipsoid that is known a priori to
its corresponding agent. Next, we present an algorithm which, by
leveraging the latter key geometric property, allows each agent to
independently compute its own cell from the HQVP. The proposed
partitioning algorithm executes a certain number of line searches
that seek for the boundary points of the cell of an agent. 
In contrast with the algorithms proposed in our previous
work~\cite{p:bakauto2014,p:bakaut2015,p:BAKOLAS2016,p:bak2018},
whose applicability is limited to partitions comprised of convex or
star convex cells, the algorithms proposed herein can successfully
characterize the cells of a HQVP despite the fact that the latter
may be non-convex or even disconnected sets. The proposed algorithms
rely on relative position measurements only and thus, neither a
global reference frame nor a common grid are required, which is in
contrast with most computational geometric techniques for
non-standard Voronoi-like partitions~\cite{p:Hoff:1999}. More
importantly, the proposed partitioning algorithm can be executed in
a distributed way (based on local information) when combined with a
network topology discovery algorithm. The main idea of the latter
algorithm is to have each agent adjust its communication range so
that it can communicate directly (point-to-point communication) with
a group of agents from the same network which is a superset of its
set of neighbors in the topology of the HQVP without having computed
the latter partition.


\textit{Structure of the paper:} The problem formulation and
corresponding preliminaries are presented in Section~\ref{s:form}.
In Section~\ref{s:partition}, we analyze the partitioning problem
and present certain key properties enjoyed by its solution. The
distributed partitioning algorithm is presented in
Section~\ref{s:mainalgo} whereas the network topology discovery
problem is analyzed and solved in Section~\ref{s:netwtopo}.
Section~\ref{s:simu} presents numerical simulations, and finally,
Section~\ref{s:concl} concludes the paper with a summary of remarks
together with directions for future work.

\vspace{-3mm}

\section{Preliminaries and Problem Formulation}\label{s:form}

\vspace{-2mm}

\subsection{Notation}
We denote by $\mathbb{R}^n$ the set of $n$-dimensional real vectors
and by $\mathbb{R}_{\geq 0}$ the set of non-negative real numbers.
We write $\mathbb{Z}$ to denote the set of integers. Given $\tau_1$,
$\tau_2 \in \mathbb{Z}$ with $\tau_1 \leq \tau_2$, we define the
\textit{discrete interval} from $\tau_1$ to $\tau_2$ as follows:
$[\tau_1, \tau_2]_{\mathbb{Z}} = [\tau_1, \tau_2] \cap \mathbb{Z}$.
We write $| \alpha|$ to denote the 2-norm of a vector
$\alpha\in\mathbb{R}^n$. Moreover, we write $\mathbf{A}\succ
\mathbf{0}$ to denote that a symmetric matrix $\mathbf{A} =
\mathbf{A}\t$ is positive definite. 
Given $\mathbf{A}=\mathbf{A}\t$, $\mathbf{B} = \mathbf{B}\t$, we
write $\mathbf{A} \succ \mathbf{B}$ if and only if $\mathbf{A} -
\mathbf{B} \succ \mathbf{0}$. Furthermore, given a symmetric matrix
$\mathbf{P} = \mathbf{P}\t$, we denote by
$\lambda_{\min}(\mathbf{P})$ and $\lambda_{\max}(\mathbf{P})$ its
minimum and maximum (real) eigenvalues, respectively. Given $x\in \mathbb{R}^n$, $\mathbf{\Sigma}  \succ \mathbf{0}$, and
$\gamma
> 0$, we write $\cE_{\gamma}(x;\mathbf{\Sigma}^{-1})$ to denote the
ellipsoid $\{z\in\mathbb{R}^n: (z - x)\t \mathbf{\Sigma} (z -x) \leq
\gamma\}$. We denote by $\cB_{\rho}(x_c)$ the closed ball of radius
$\rho>0$ centered at $x_c$, that is, $\cB_{\rho}(x_c):=
\{z\in\mathbb{R}^n: |z - x_c | \leq \rho\}$. Furthermore,
$\mathrm{bd}(\cA)$ and $\mathrm{rbd}(\cA)$ denote the boundary and
the relative boundary of a set $\cA$, whereas $\intr(\cA)$ and
$\mathrm{rint}(\cA)$ denote its interior and relative interior. The powerset of a set $\cA$ is denoted as $\wp(\cA)$. 
Given $\cA$, $\cB \subseteq \mathbb{R}^n$, we denote by $\cA \oplus \cB$ their Minkowski sum, that
is, $\cA \oplus \cB := \{x=y+z: y\in \cA~\text{and}~z\in\cB \}$, and
by $\cA \ominus \cB$ their Minkowski difference, that is, $\cA
\ominus \cB := \{ x: \{x\} \oplus \cB \subseteq \cA \}$. Given $\alpha$, $\beta\in\mathbb{R}^n$, we denote by
$[\alpha,\beta]$ the line segment connecting them (including the two
endpoints), that is, $[\alpha,\beta] :=
\{\x\in\mathbb{R}^n:~\x=t\alpha + (1-t)\beta,~t\in[0,1] \}$. In
addition, we denote by $]\alpha,\beta]$ and $[\alpha,\beta[$ the
sets $[\alpha,\beta]\backslash\{ \alpha \}$ and
$[\alpha,\beta]\backslash\{\beta\}$, respectively.

\subsection{The Partitioning Problem for a Heterogeneous Multi-Agent Network}

In this section, we formulate the partitioning problem for a
multi-agent network comprised of $n$ agents distributed over a
spatial domain $\cS$, which is assumed to be a convex and compact
set. To the latter network we attach an additional agent, which we
refer to as the $0$-th agent of the network. The latter agent may
correspond, for instance, to a vehicle station from which vehicles
are dispatched in response to requests issued in the vicinity of the
station or a ``mother vehicle'' that can deploy $n$ mobile sensors
to collect measurements from various nearby locations.
We will refer to the network that includes the $0$-th
agent as the \textit{extended} network. It is assumed that the
agents are located at $n+1$ distinct locations in $\cS$, which form
the point-set $X:=\{x_i\in\cS:~i\in[0,n]_{\mZ}\}$.


Our first objective is to subdivide $\cS$ into $n+1$ non-overlapping
subsets that will be associated with the $n+1$ agents of the
extended network in an one-to-one way. We will refer to these
subsets of $\cS$ as regions of influence (ROI) or simply cells that
comprise a spatial partition of the network's workspace. In
particular, the interior of each cell will consist exclusively of
points in $\cS$ that are ``closer'' to its corresponding agent than
to any other agent of the extended network. The closeness between
the $i$-th agent and an arbitrary point $x \in \cS$ will be measured
in terms of an appropriate convex quadratic \textit{proximity}
(generalized) metric $\delta(\cdot;x_i):\cS \rightarrow
\mathbb{R}_{\geq 0}$ with
\begin{equation}\label{eq:deltadef}
 \delta_i(x;x_i) \defn  (x - x_i)\t \mathbf{P}_i (x - x_i) +
 \mu_i,
\end{equation}
where $\mu_i \geq 0$ and $\mathbf{P}_i \succ \mathbf{0}$ for all
$i\in [0,n]_{\mathbb{Z}}$. We will refer to $\mu_i$ and
$\mathbf{P}_i$ as the $i$-th additive gain and distance operator,
respectively. The proximity metric $\delta_i(x;x_i)$ corresponds,
for instance, to the cost that the $i$-th agent will incur for its
transition from point $x_i$ to point $x$. 
Alternatively, it may reflect the sensing cost that the $i$-th
agent, which is located at $x_i$, will incur in order to obtain
measurements from point $x$. In particular, let us consider the bivariate Gaussian distribution with mean $m_i \in
\mathbb{R}^2$ and covariance $\mathbf{\Sigma}_i \succ \mathbf{0}$ whose probability density function is given by
\[
\rho_i(x) : = \big(2\pi \sqrt{\det(\mathbf{\Sigma}_i)} \big)^{-1}
\mathrm{exp}\big(-\tfrac{1}{2} (x-m_i)\t
\mathbf{\Sigma}^{-1}_i(x-m_i) \big)
\]
and let us define the sensing cost as
follows~\cite{p:arslan2019}:
\begin{align*}
c_i(x) & := -\mathrm{log}(\rho_i(x)) \nonumber \\
 & = \log\big(2\pi
\sqrt{\det(\mathbf{\Sigma}_i)} \big) 
+ \tfrac{1}{2}(x-m_i)\t \mathbf{\Sigma}^{-1}_i(x-m_i).
\end{align*}
Therefore, by taking $\mathbf{P}_i := \tfrac{1}{2}
\mathbf{\Sigma}^{-1}_i$, $x_i=m_i$ and $\mu_i:= \log(2\pi
\sqrt{\det(\mathbf{\Sigma}_i)}$, we have $\delta_i(x;x_i) = c_i(x)$.

It is worth noting that the $i$-th additive gain $\mu_i$ corresponds
to the minimum value of $\delta_i(x;x_i)$, which is attained at
$x=x_i$, that is, $\mu_i = \min_{x\in\cS} \delta_i(x;x_i) =
\delta_i(x_i;x_i)$. In addition, the $i$-th distance operator
$\mathbf{P}_i$ determines which directions, if any, are more
preferable to the $i$-th agent than others. In particular, if
$\mathbf{P}_i=\lambda_i \mathbf{I}$, where $\lambda_i
>0$, then the level sets of the quadratic form $(x - x_i)\t
\mathbf{P}_i (x - x_i)$ are circles and thus there are no preferable
directions; otherwise, the latter level sets become ellipses whose
major axes determine the most preferable directions. In the first
case, $\mathbf{P}_i$ is an isotropic distance operator (i.e.,
direction independent), whereas in the second, and more interesting
case, is an anisotropic (i.e., direction-dependent) distance
operator. It is worth noting that requiring the existence of
a matrix $\overline{\mathbf{P}} \succ \mathbf{0}$ such that
$\mathbf{P}_i = \overline{\mathbf{P}}$ for all $i \in
[0,n]_{\mathbb{Z}}$ can be a very restrictive assumption in
practice. In this work, we will consider the more general case in
which there always exists $(i,j)$ with $i\neq j$ such that
$\mathbf{P}_i \neq \mathbf{P}_j$ and we will refer to the
multi-agent network as ``heterogeneous.''

Next, we provide a number of technical, yet practically intuitive,
assumptions that will help us streamline the subsequent discussion
and analysis.
\begin{assumption}\label{assumption1}
For any $i \in [0,n]_{\mathbb{Z}}$, we have that $\delta_i(x_i;x_i)
<\delta_j(x_i;x_j)$ or, equivalently, 
\begin{equation}\label{eq:assu1}
(x_j - x_i)\t \mathbf{P}_j (x_j - x_i) + \mu_j > \mu_i,
\end{equation}
for all $j \neq i$, provided that $x_i \neq x_j$.
\end{assumption}
The previous assumption implies that the distance of the $j$-th
agent from the location $x_i$ of the $i$-th agent, which is equal to
$\delta_j(x_i;x_j)$, has to be greater than the distance of the
$i$-th agent from itself, which is equal to $\delta_i(x_i;x_i) =
\mu_i$. For instance, in the case of a sensor network, condition
\eqref{eq:assu1} implies that no sensor different from the $i$-th
sensor can obtain more accurate measurements from the location $x_i$
of the $i$-th agent. 

\begin{remark}
Although Assumption~\ref{assumption1} is quite intuitive, one may argue that there may exist applications in which it may not hold true. It should be mentioned here that the partitioning algorithm that will be presented herein can be applied even when Assumption~1 is removed, after the necessary modifications have been carried out (we will comment on some of these modifications later on). Assumption~\ref{assumption1} will allow us to streamline the presentation and avoid discussing special cases of low interest.
\end{remark}

\begin{assumption}\label{assumption2}
We assume that
\begin{equation}\label{eq:partialorder}
 \mathbf{P}_i \succ \mathbf{P}_0 \succ \mathbf{0}, ~~~~ \mu_i \geq \mu_0 \geq 0,~~~~\forall i \in
 [1,n]_{\mZ}.
\end{equation}
\end{assumption}

The following proposition will allow us to better understand the
implications of Assumption~\ref{assumption2}.

\begin{proposition}\label{pr:setcontain}
Let $\gamma > \max_{i\in[1,n]_{\mZ}}\mu_i$ and let $\mx \in \cS$. In
addition, let $\cD^0_{\gamma}(\mx)$ and $\cD^i_{\gamma}(\mx)$ denote
the $\gamma$- sublevel-sets of, respectively, $\delta_0(\cdot;x_0)$
and $\delta_i(\cdot;x_i)$ when $x_i\equiv x_0 \equiv \mx$ for all
$i\in[1,n]_{\mZ}$, that is, $\cD_\gamma^0(\mx):= \{ x\in
\cS:~\delta_0(x;\mx) \leq \gamma \}$ and $\cD_{\gamma}^i(\mx):= \{
x\in \cS:~\delta_i(x;\mx) \leq \gamma\}$, for $i \in [1,n]_{\mZ}$.
Then, the following set inclusion holds:
\begin{equation}\label{eq:setcontain}
\cD_{\gamma}^i(\mx) \subsetneq \cD_{\gamma}^0(\mx),~~~\forall i \in
[1,n]_{\mZ}.
\end{equation}
\end{proposition}
\begin{proof}
In view of \eqref{eq:deltadef}, $\cD_{\gamma}^0(\mx)$ and
$\cD_{\gamma}^i(\mx)$ can be expressed as follows:
\begin{align*}
\cD_{\gamma}^0(\mx) & = \{ x\in \cS :~(x- \mx)\t \mathbf{P}_0(x -
\mx)
\leq \gamma  - \mu_0\}, \\
\cD_{\gamma}^i(\mx) & = \{ x\in \cS:~(x- \mx)\t \mathbf{P}_i(x -
\mx) \leq \gamma - \mu_i\}.
\end{align*}
By hypothesis
$\gamma
> \mu_i \geq \mu_0 \geq 0$, and thus
\begin{align*}
\cD_{\gamma}^0(\mx) & \supseteq \{ x\in \cS
:~(x- \mx)\t \mathbf{P}_0(x - \mx) \leq \gamma  - \mu_i\} \\
& \supsetneq \{ x\in \cS:~(x- \mx)\t \mathbf{P}_i(x - \mx) \leq
\gamma - \mu_i\} \\ & = \cD_{\gamma}^i(\mx),
\end{align*}
where the second set inclusion follows from the fact that
$\mathbf{P}_i \succ \mathbf{P}_0 \succ \mathbf{0}$. Thus, the set
inclusion \eqref{eq:setcontain} holds true.
\end{proof}

\begin{remark}
It is worth noting that $\cD_{\gamma}^0(\mx) = \cE_{\gamma -
\mu_0}(\mx;\mathbf{P}_0^{-1}) \cap \cS$ and $\cD_{\gamma}^i(\mx) =
\cE_{\gamma - \mu_i}(\mx;\mathbf{P}_i^{-1}) \cap \cS$.
Proposition~\ref{pr:setcontain} implies that the footprint of the
set of points that are within distance $\gamma$ from the $0$-th
agent (distance measured in terms of $\delta_0$) is greater than the
footprint of the set of points that are within distance $\gamma$
from the $i$-th agent (distance measured now in terms of $\delta_i$)
when both of the agents are placed at an arbitrary common point $\mx
\in \cS$.
\end{remark}

\subsection{Formulation of the Workspace Partitioning Problem}

We can now give the precise definitions of the Voronoi-like
partition of $\cS$ generated by the extended multi-agent network
based on the quadratic proximity metrics defined in
\eqref{eq:deltadef}.
\begin{definition}\label{defn:QVP}
Suppose that $\cS \in \mathbb{R}^2$ is a compact and convex set and
let $X \subset \cS$ be a set comprised of $n+1$ distinct points
(locations of the agents). Then, we say that the collection of sets
$\cV(X;\cS) := \{ \cV^i \in \wp(\cS):~ i\in [0, n]_{\mZ} \}$ where
\begin{equation}\label{eq:voronoi}
\cV^i:=\{x \in \cS:~\delta_i(x; x_i) \leq
\min_{j\neq i} \delta_j(x;x_j)\},
\end{equation}
forms a Heterogeneous Quadratic Voronoi Partition (HQVP) of $\cS$
that is generated by $X$. In particular, $\mathrm{i})$ $\cS =
\cup_{i \in [0,n]_{\mathbb{Z}} } \cV^i$ and $\mathrm{ii})$ $\intr(\cV^i)
\cap \intr(\cV^j) = \varnothing$, for $i\neq j$. We will refer to
the set $\cV^i$ as the $i$-th cell or region-of-influence (ROI).
\end{definition}
The following proposition highlights some fundamental properties of
the HQVP.
\begin{proposition}\label{prop:basicprop}
Let $\cV^i \in \cV(X;\cS)$. Then, $\delta_i(x; x_i) \leq \min_{j\neq
i} \delta_j(x;x_j)$ for all $x \in \cV^i$ and in particular,
\begin{enumerate}
\item
 $\delta_i(x; x_i) < \min_{j\neq i} \delta_j(x;x_j),~~\forall x \in \intr(\cV^i)$ 
\item $\delta_i(x; x_i) = \min_{j\neq i} \delta_j(x;x_j),~~\forall x\in \mathrm{bd}(\cV^i) \backslash \mathrm{bd}(\cS)$, that is,
there exists $j=j_x$ such that $\delta_i(x; x_i) =
\delta_{j_x}(x;x_{j_x})$.
\end{enumerate}
\end{proposition}

It is worth considering what would happen if we dropped
Assumption~\ref{assumption2} and assumed instead that $\mu_i =
\bar{\mu}$ and $\mathbf{P}_i= \lambda \mathbf{I}$, for all
$i\in[0,n]_{\mZ}$, where $\bar{\mu} \geq 0$ and $\lambda
>0$. In this special case, each agent employs the same proximity metric; in
particular, $\delta_i(x; x_i) = \lambda
 |x-x_i|^2 + \bar{\mu}$, for all $i \in [0,n]_{\mZ}$. In this case,
\begin{align*}
\cV^i & :=\{x \in \cS:~ \lambda
 |x-x_i|^2 \leq \lambda \min_{j\in[0,n]_{\mZ}}
 |x-x_j|^2\}\\
&=\{x \in \cS:~
 |x-x_i| \leq \min_{j\in[0,n]_{\mZ}} |x-x_j|\},
\end{align*}
which is precisely the definition of the $i$-th cell of the
\textit{standard} Voronoi partition~\cite{p:voronoi}. Consequently,
in this special case, the HQVP reduces to the standard Voronoi
partition which has combinatorial complexity in $\mathcal{O}(n)$ and
computational complexity in $\cO(n \log(n))$. Another special case
while keeping Assumption~\ref{assumption2} inactive, is when there
is a pair $(i,j)$, with $i \neq j$, such that $\mu_i \neq \mu_j$ and
$\mathbf{P}_i = \overline{\mathbf{P}}$, for all $i\in [0, n]_{\mZ}$,
where $\overline{\mathbf{P}} \succ \mathbf{0}$. As we have shown
in~\cite{p:bakolas2013b}, the HQVP in the latter case reduces to an
\textit{affine diagram}, which has combinatorial complexity in
$\Theta(n)$ and computational complexity in $\Theta(n \log n +
n)$~\cite{b:boigeom} (note that the latter complexities are modest
and close to those of the standard Voronoi partition). In this work,
in view of Assumption~\ref{assumption2}, there always exists a pair
$(i,j)$, with $i\neq j$, such that $\mathbf{P}_i \neq \mathbf{P}_j$
(one can take $j=0$ and any $i\in [1,n]_{\mZ}$). According
to~\cite{p:BoissoAnisoVor2008}, the HQVP has combinatorial
complexity $\Theta(n^3)$ and computational complexity in $\cO(n^3 +
n\log(n))$; these complexities are significantly higher than those
of the standard and the affine Voronoi partitions. One important
fact is that the cells of HQVP are not necessarily convex sets (they
may even be disconnected sets), which makes their computation by
means of distributed algorithms quite challenging. By virtue of the previous discussion, it should become clear that the partitioning algorithms proposed in our previous work~\cite{p:bakauto2014,p:bakaut2015,p:BAKOLAS2016,p:bak2018}, which can only compute affine partitions or partitions comprised of star convex cells for homogeneous multi-agent networks, are not applicable to the partitioning problem for heterogeneous networks which is considered herein. The latter problem requires the development of new and more powerful tools which are applicable to partitions comprised of cells which can be non-convex or even disconnected sets.

Next, we formulate the uncoupled partitioning problem in which the
$i$-th agent of the network is required to compute its own cell in
HQVP independently from its teammates.
\begin{problem}\label{problemDP}
\textit{Uncoupled Partitioning Problem over $\cS$:}~Let $\cV(X;\cS)
=\{\cV_i\in \wp(\cS):~i\in[0,n]_{\mZ}\}$ be the HQVP of $\cS$
generated by the point-set $X : = \{ x_i\in\cS:~i\in[0,n]_{\mZ}\}$.
For a given $i \in [1,n]_{\mZ}$, compute the cell $\cV^i \in
\cV(X;\cS)$, independently from the other cells of the same
partition.
\end{problem}

\begin{remark}
It is worth noting that the computation of the cell $\cV^0$ which is
assigned to the $0$-th agent of the extended network is not included
in the formulation of Problem~\ref{problemDP}. The latter set
corresponds to the part of the spatial domain $\cS$ that is not
claimed by any agent of the actual network or in other words, the
coverage hole of the latter network, that is, $\cV^0 = \cS
\backslash \big( \cup_{i=1}^{n} \cV^i \big)$. Intuitively, this means that at any point in $\cV^0$, the ground station or mother vehicle (the latter correspond to interpretations of the hypothetical $0$-th agent) can rely to their own sensing capabilities and therefore, they do not have to dispatch any mobile sensors from the actual network to take in-situ measurements there. Note that the non-emptiness of the coverage hole $\cV^0$ is a
direct consequence of Assumption~\ref{assumption2}. 
\end{remark}

\subsection{Formulation of the Network Topology Discovery Problem}

In a nutshell, the goal of the network topology discovery problem is
to find a systematic way that will allow the $i$-th agent of the
network to determine its neighbors in the topology induced by the
HQVP.

\begin{definition}\label{def:Topology}
The $i$-th agent and the $j$-th agent, which are located at $x_i\in
X$ and $x_j\in X$, respectively, are \textit{neighbors} in the topology of
$\cV(X;\cS)$, if the boundaries of their cells have a non-empty
intersection, that is, $\mathrm{bd}(\cV^i) \cap \mathrm{bd}(\cV^j)
\neq  \varnothing$.
\end{definition}

Now, let us denote by $\cN_i$ the index set of the neighbors of the
$i$-th agent. In view of Definition~\ref{def:Topology},
\begin{equation}\label{eq:Ni}
\cN_i := \{ \ell \in [0,n]_{\mZ} \backslash \{i\}:
\mathrm{bd}(\cV^{\ell}) \cap \mathrm{bd}(\cV^i) \neq \varnothing \}.
\end{equation}
\begin{proposition}
The index-set of the neighbors of the $i$-th agent, $\cN_i$,
consists of all $\ell \in [0,n]_{\mZ} \backslash \{i\}$ such that
$\delta_{\ell}(x;x_{\ell}) = \delta_i(x;x_i)$ for some $x \in
\mathrm{bd}(\cV^i) \backslash \mathrm{bd}(\cS)$.
\end{proposition}
\begin{proof}
The proof follows readily from Proposition~\ref{prop:basicprop}.
\end{proof}

The network topology discovery problem seeks for a lower bound on
the communication range $\eta_i$ of the $i$-th agent such that its
communication region $\cB_{\eta_i}(x_i)$ contains all of its
neighbors in the topology of HQVP. 
\begin{problem}
\textit{Network Topology Discovery Problem:} Find a lower bound
$\underline{\eta_i} > 0$ on the communication range $\eta_i$ of the
$i$-th agent, for $i \in [1,n]_{\mZ}$, such that its communication
region, $\cB_{\eta_i}(x_i)$, contains all of its neighbors, that is,
\begin{equation}
\cB_{\eta_i}(x_i) \supsetneq \{ x_k \in X:~k\in \cN_i \},~~\forall
\eta_i \geq \underline{\eta_i}.
\end{equation}
\end{problem}

\section{Analysis and Solution of the Uncoupled Partitioning
Problem}\label{s:partition}

\subsection{Analysis of the Uncoupled Partitioning Problem}
In this section, we will present some useful properties enjoyed by
the cells comprising the HQVP which we will subsequently leverage to
develop distributed algorithms for the computation of the solution
to Problem~\ref{problemDP}. The first step of our analysis will be
the characterization of the bisector, $\mB_{i,j}$, that corresponds
to the loci of all points in $\cS$ that are equidistant from the
$i$-th and the $j$-th agents with $i \neq j$, that is,
\begin{equation}\label{eq:bisector0}
\mB_{i,j} := \{ x\in\cS:~\delta_i(x;x_i) = \delta_j(x;x_j)\}.
\end{equation}
The equation $\delta_i(x;x_i) = \delta_j(x;x_j)$ is equivalent to
\begin{align*}
(x - x_i)\t \mathbf{P}_i (x - x_i) + \mu_i = (x - x_j)\t
\mathbf{P}_j (x - x_j) + \mu_j
\end{align*}
which can be written more compactly as follows
\begin{align}\label{eq:bisector}
 x\t \mathbf{P}_{i,j} x - 2\chi_{i,j}\t x + \sigma_{i,j}=0,
\end{align}
where
\begin{subequations}
\begin{align}
\mathbf{P}_{i,j} & := \mathbf{P}_i -\mathbf{P}_j, \label{eq:Pij} \\
\chi_{i,j} & := \mathbf{P}_i x_i -\mathbf{P}_j x_j, \label{eq:chiij} \\
\sigma_{i,j} & := |\mathbf{P}_i^{1/2} x_i|^2 + \mu_i
-|\mathbf{P}^{1/2}_j x_j|^2 - \mu_j. \label{eq:sigmaij}
\end{align}
\end{subequations}
If $\mathbf{P}_{i,j}=\mathbf{0}$, that is, $\mathbf{P}_i =
\mathbf{P}_j$, equation \eqref{eq:bisector} describes a straight
line. In the more interesting case when $\mathbf{P}_{i,j} \neq
\mathbf{0}$, \eqref{eq:bisector} corresponds to a quadratic vector
equation that determines a conic section.

Next, we will leverage Assumption~\ref{assumption2} to show that the
cell $\cV^i$, for $i \in [1,n]_{\mZ}$, enjoys an important property
that will prove very useful in our subsequent analysis. To this aim,
we first note that, in view of Assumption~\ref{assumption2},
$\mathbf{P}_i \succ \mathbf{P}_0$ or equivalently
$\mathbf{P}_{i,0}\succ \mathbf{0}$. Next, by completing the square
in \eqref{eq:bisector} and then setting $j=0$, we get
\begin{align*}
0 & = x\t \mathbf{P}_{i,0} x - 2\chi_{i,0}\t \mathbf{P}_{i,0}^{-1/2}
\mathbf{P}_{i,0}^{1/2} x + \chi_{i,0}\t \mathbf{P}_{i,0}^{-1}
\chi_{i,0} \\
&~~~ - \chi_{i,0}\t \mathbf{P}_{i,0}^{-1} \chi_{i,0} + \sigma_{i,0}
\end{align*}
from which it follows that
\begin{equation}\label{eq:bisector2}
|\mathbf{P}_{i,0}^{1/2} (x - \mathbf{P}_{i,0}^{-1}  \chi_{i,0})|^2 =
|\mathbf{P}_{i,0}^{-1/2} \chi_{i,0}|^2 - \sigma_{i,0}.
\end{equation}
Therefore, the bisector $\mB_{i,0}$ consists of all points $x \in
\cS$ that satisfy Eq.~\eqref{eq:bisector2}, which is the equation of
an ellipse provided that the right hand side of the latter equation
is a strictly positive number.

\begin{proposition}\label{prop:baslemma}
Let $i \in [1,n]_{\mZ}$ and let
\begin{equation}\label{eq:ellij}
\ell_{i,0}:=| \mathbf{P}_{i,0}^{-1/2} \chi_{i,0}|^2 - \sigma_{i,0},
\end{equation}
where $\mathbf{P}_{i,0}$, $\chi_{i,0}$ and $\sigma_{i,0}$ are as
defined in \eqref{eq:Pij}--\eqref{eq:sigmaij} for $j=0$. Then, $\ell_{i,0}>0$
and the bisector $\mB_{i,0}$ satisfies 
\begin{equation}\label{bisdefn}
\mB_{i,0} = \mathrm{bd}(E_i) \cap \cS,
\end{equation}
where $E_i := \mathcal{E}_{\ell_{i,0}}(\mathbf{P}_{i,0}^{-1}
\chi_{i,0}; \mathbf{P}_{i,0}^{-1} )$.
\end{proposition}
\begin{proof}
In view of \eqref{eq:Pij}-\eqref{eq:chiij} for $j=0$, we have
\begin{align*}
|\mathbf{P}_{i,0}^{-1/2} \chi_{i,0}|^2  &= |\mathbf{P}_{i,0}^{-1/2}
(\mathbf{P}_i x_i -\mathbf{P}_0 x_0)|^2\\
& = x_i\t\mathbf{P}_i\mathbf{P}_{i,0}^{-1}\mathbf{P}_i x_i +
x_0\t\mathbf{P}_0\mathbf{P}_{i,0}^{-1}\mathbf{P}_0 x_0 \nonumber \\
&~~~~ - 2 x_i\t\mathbf{P}_i\mathbf{P}_{i,0}^{-1}\mathbf{P}_0 x_0
\nonumber \\
& = [x_i\t,~x_0\t] \begin{bmatrix}
\mathbf{P}_i\mathbf{P}_{i,0}^{-1}\mathbf{P}_i &
-\mathbf{P}_i\mathbf{P}_{i,0}^{-1}\mathbf{P}_0 \\
-\mathbf{P}_0\mathbf{P}_{i,0}^{-1}\mathbf{P}_i &
\mathbf{P}_0\mathbf{P}_{i,0}^{-1}\mathbf{P}_0
\end{bmatrix} \begin{bmatrix} x_i \\ x_0 \end{bmatrix}.
\end{align*}
In addition, from \eqref{eq:sigmaij} for $j=0$, we get
\begin{align*}
\sigma_{i,0} & = |\mathbf{P}_i^{1/2} x_i|^2 +
\mu_i -| \mathbf{P}^{1/2}_0 x_0 |^2 - \mu_0\\
& = x_i\t \mathbf{P}_i x_i - x_0\t \mathbf{P}_0 x_0 + \mu_i - \mu_0 \\
& = [x_i\t,~x_0\t] \begin{bmatrix}
\mathbf{P}_i  & \mathbf{0}\\
\mathbf{0} & -\mathbf{P}_0
\end{bmatrix} \begin{bmatrix}x_i \\ x_0 \end{bmatrix} + \mu_i - \mu_0.
\end{align*}
Therefore, we have that
\begin{align}
\ell_{i,0} & = |\mathbf{P}_{i,0}^{-1/2} \chi_{i,0}|^2 - \sigma_{i,0} \nonumber \\
& = [x_i\t,~x_0\t] \begin{bmatrix}
\mathbf{P}_i\mathbf{P}_{i,0}^{-1}\mathbf{P}_i - \mathbf{P}_i &
-\mathbf{P}_i \mathbf{P}_{i,0}^{-1}\mathbf{P}_0 \\
-\mathbf{P}_0 \mathbf{P}_{i,0}^{-1}\mathbf{P}_i & \mathbf{P}_0
\mathbf{P}_{i,0}^{-1}\mathbf{P}_0 + \mathbf{P}_0
\end{bmatrix} \begin{bmatrix}x_i \\ x_0\end{bmatrix} \nonumber \\
&~~~ +\mu_0 - \mu_i.\label{eq:intproof1}
\end{align}
Now, in view of \eqref{eq:assu1} for $j=0$, we have that
\begin{align}\label{eq:helpeq1}
\mu_0 - \mu_i & > - (x_0 -x_i)\t\mathbf{P}_{0} (x_0 - x_i) \nonumber \\
&= [x_i\t,~~x_0\t] \begin{bmatrix}
-\mathbf{P}_0  & \mathbf{P}_0\\
\mathbf{P}_0 & -\mathbf{P}_0
\end{bmatrix} \begin{bmatrix}x_i\\x_0\end{bmatrix}.
\end{align}
Therefore, in view of \eqref{eq:intproof1}, \eqref{eq:helpeq1} gives
\begin{equation}\label{eq:lijproof}
\ell_{i,0}  > [x_i\t,~x_0\t] \mathbf{\Psi}
\begin{bmatrix}x_i\\x_0\end{bmatrix},~~~~~~\mathbf{\Psi} := \begin{bmatrix} \mathbf{\Psi}_{11} & \mathbf{\Psi}_{12} \\
\mathbf{\Psi}_{12}\t & \mathbf{\Psi}_{22} \end{bmatrix},
\end{equation}
where $\mathbf{\Psi}_{11}, \mathbf{\Psi}_{12}$, $\mathbf{\Psi}_{13}
\in \mathbb{R}^{2\times2}$ are defined as follows:
\begin{subequations}
\begin{align}\label{eq:Psi}
\mathbf{\Psi}_{11} & :=
\mathbf{P}_i\mathbf{P}_{i,0}^{-1}\mathbf{P}_i
- \mathbf{P}_i - \mathbf{P}_0, \\
\mathbf{\Psi}_{12} & :=
-\mathbf{P}_i\mathbf{P}_{i,0}^{-1}\mathbf{P}_0 + \mathbf{P}_0, \\
\mathbf{\Psi}_{22} & := \mathbf{P}_0 \mathbf{P}_{i,0}^{-1}
\mathbf{P}_0. \label{eq:Psi3}
\end{align}
\end{subequations}
Note that $\mathbf{\Psi}_{22} =
\mathbf{P}_0\mathbf{P}_{i,0}^{-1}\mathbf{P}_0 \succ \mathbf{0}$.
Next, we show that the Schur complement of the block
$\mathbf{\Psi}_{22}$ of the block matrix $\mathbf{\Psi}$, which is
denoted as $(\mathbf{\Psi}/ \mathbf{\Psi}_{22})$ and defined as
$(\mathbf{\Psi}/ \mathbf{\Psi}_{22}) :=\mathbf{\Psi}_{11} -
\mathbf{\Psi}_{12} \mathbf{\Psi}_{22}^{-1} \mathbf{\Psi}_{12}\t$, is
positive definite, that is, $(\mathbf{\Psi}/ \mathbf{\Psi}_{22})
\succ \mathbf{0}$. Indeed, in view of \eqref{eq:Psi}-\eqref{eq:Psi3}
\begin{align}\label{eq:SchurPsi}
 (\mathbf{\Psi}/ \mathbf{\Psi}_{22})
& = \mathbf{P}_i \mathbf{P}_{i,0}^{-1} \mathbf{P}_i - \mathbf{P}_i,
\end{align}
where $\mathbf{P}_{i,0} = \mathbf{P}_i - \mathbf{P}_0$. Furthermore,
in light of \eqref{eq:partialorder}, we have that $\mathbf{0} \prec
\mathbf{P}_{i,0} = \mathbf{P}_i - \mathbf{P}_0 \prec \mathbf{P}_i$
which implies that $\mathbf{0} \prec \mathbf{P}^{-1}_i \prec
\mathbf{P}_{i,0}^{-1}$ and thus
\begin{equation}\label{eq:bPSDineq0}
\mathbf{I} \prec \mathbf{P}^{1/2}_i \mathbf{P}_{i,0}^{-1}
\mathbf{P}^{1/2}_i.
\end{equation}
After pre- and post-multiply \eqref{eq:bPSDineq0} with
$\mathbf{P}^{1/2}_i$, we take
\begin{equation}\label{eq:bPSDineq}
\mathbf{P}_i \mathbf{P}_{i,0}^{-1} \mathbf{P}_i \succ \mathbf{P}_i.
\end{equation}
In view of \eqref{eq:bPSDineq}, \eqref{eq:SchurPsi} implies that
$(\mathbf{\Psi}/ \mathbf{\Psi}_{22}) \succ \mathbf{0}$. The fact
that $(\mathbf{\Psi}/ \mathbf{\Psi}_{22}) \succ \mathbf{0}$ and
$\mathbf{\Psi}_{22} \succ \mathbf{0}$ imply that $\mathbf{\Psi}
\succ \mathbf{0}$. Consequently, by virtue of \eqref{eq:lijproof},
we take $\ell_{i,0} > 0$, for all $i\in[1,n]_{\mZ}$. Then, all
points $x\in\cS$ that satisfy \eqref{eq:bisector2} belong to the
boundary of the ellipsoid $E_i$, and thus $\mB_{i,0} \subseteq
\mathrm{bd}(E_i) \cap \cS$. The set inclusion $\mB_{i,0} \supseteq
\mathrm{bd}(E_i) \cap \cS$ can be shown similarly and thus, equation
\eqref{bisdefn} follows readily. The proof is now complete.
\end{proof}

\begin{proposition}\label{prop:main}
Let $i\in[1,n]_{\mZ}$ and let
$E_i:=\mathcal{E}_{\ell_{i,0}}(\mathbf{P}_{i,0}^{-1} \chi_{i,0};
\mathbf{P}_{i,0}^{-1} )$. Then, the cell $\cV^i \in\cV(X;\cS)$
satisfies the following set inclusion:
\begin{equation}\label{eq:mainsetincl}
\cV^i \subseteq E_i \cap \cS.
\end{equation}
\end{proposition}
\begin{proof}
Let us consider the two disjoint sets $\cS_{i,0}:= \{ x\in\cS:
\delta_i(x;x_i) \leq \delta_0(x;x_0) \}$ and $\cS^c_{i,0} := \{
x\in\cS: \delta_i(x;x_i)
> \delta_0(x;x_0) \}$ whose union is equal to $\cS$. 
By definition,
\begin{align}\label{eq:profSi0}
\cS_{i,0} \supseteq \{ x\in\cS: \delta_i(x;x_i) \leq
\min_{j\in[0,n]_{\mZ}}\delta_j(x;x_j) \} = \cV^i,
\end{align}
where the last set equality follows from \eqref{eq:voronoi}. Next,
we show that $\cS_{i,0} = E_i \cap \cS$. Indeed, let $x \in E_i\cap
\cS$. Then, in view of \eqref{eq:bisector2} and \eqref{eq:ellij}, we
have that
\begin{equation}
|\mathbf{P}_{i,0}^{1/2} (x - \mathbf{P}_{i,0}^{-1} \chi_{i,0})|^2
\leq \ell_{i,0},
\end{equation}
which implies, after following backwards the derivation from
\eqref{eq:bisector0}--\eqref{eq:sigmaij} for $j=0$, that
\[ (x - x_i)\t \mathbf{P}_i (x - x_i) + \mu_i \leq (x -
x_0)\t \mathbf{P}_j (x - x_0) + \mu_0
\]
which proves that $x \in \cS_{i,0}$ and thus $\cS_{i,0} \subseteq
E_i \cap \cS$. The set inclusion $E_i \cap \cS \subseteq \cS_{i,0}$
can be proven similarly. Therefore, $\cS_{i,0} = E_i \cap \cS$ and
thus, in view of \eqref{eq:profSi0}, we conclude that $\cV^i
\subseteq E_i \cap \cS$ which completes the proof.
\end{proof}

\begin{remark}
Proposition~\ref{prop:main} implies that the $i$-th agent can
determine the compact and convex set $E_{i} \cap \cS$ that will
necessarily contain its cell $\cV^i$ provided that the quantities
$\mathbf{P}_0$, $\mu_0$, and $x_0$, which are associated with the
$0$-th agent of the extended network, are known to it. All of these
quantities can be determined by the agents of the actual network by
means of distributed algorithms. For instance, $x_0$ can be taken to
be the average position of the agents of the actual network and thus
can be computed by means of standard average consensus
algorithms~\cite{p:consepropa2006,p:XIAO2007}. In addition, we can
set $\mu_0 := \min\{\mu_i\in \mathbb{R}_{\geq
0}:~i\in[1,n]_{\mZ}\}$, which is in accordance with
Assumption~\ref{assumption2} and can be computed by means of, for
instance, the flooding algorithm which is one of the simplest
distributed algorithms~\cite{b:lynch1996}. Furthermore, we can take
$\mathbf{P}_0 = \lambda_0 I$, where $0<\lambda_0 < \min\{
\lambda_{\min}(\mathbf{P}_i):~i\in[1,n]_{\mZ}\}$ so that
Assumption~\ref{assumption2} is respected; again, one can compute
$\lambda_0$ by means of a flooding-type distributed algorithm. 
\end{remark}

\subsection{The $i$-th lower envelope $\Delta_i$}\label{ss:param}

Let us consider the $i$-th lower envelope function
$\Delta_i(\cdot;X): \cS \rightarrow \mathbb{R}$ with
\begin{equation}\label{eq:Deltaeq}
\Delta_i(x;X) := \min_{\ell \neq i} \delta_{\ell}(x;x_{\ell}) -
\delta_i(x;x_i).
\end{equation}

\begin{proposition}\label{prop:globineq}
Let $i\in[1,n]_{\mZ}$ and let
$E_i:=\mathcal{E}_{\ell_{i,0}}(\mathbf{P}_{i,0}^{-1} \chi_{i,0};
\mathbf{P}_{i,0}^{-1} )$. Then, $x \in \cV^i$ if and only if $\Delta_i(x;X)\geq 0$, that is, 
\begin{equation}\label{eq:Vialter}
\cV^i = \{ x\in E_i \cap \cS: \Delta_i(x;X)\geq 0\}.
\end{equation}
Moreover,
\begin{subequations}
\begin{align}\label{eq:globineq}
\Delta_i(x;X) & > 0,~~~\forall x\in \mathrm{int}(\cV^i), \\
\Delta_i(x;X) & = 0,~~~\forall x\in \mathrm{bd}(\cV^i) \backslash
\mathrm{bd}(\cS). \label{eq:globineq2}
\end{align}
\end{subequations}
\end{proposition}
\begin{proof}
Equation~\eqref{eq:Vialter} follows from Definition~\ref{defn:QVP}
and Proposition~\ref{prop:main}. In addition,
\eqref{eq:globineq}--\eqref{eq:globineq2} follows from
Proposition~\ref{prop:basicprop} and \ref{prop:main}.
\end{proof}

Besides the $i$-th lower envelope, we can also define the global lower envelope function $\Delta(\cdot;X):\cS \rightarrow \mathbb{R}$ with
\begin{equation}\label{eq:Deltaglobal}
\Delta(x;X) := \min_{\ell \in [0,n]_{\mathbb{Z}}} \delta_{\ell}(x;x_{\ell}).
\end{equation}
In view of the definition of the $\cV^i$ given in \eqref{eq:voronoi}, it follows immediately that a point $x\in \cV^i$ if and only if $\delta_i(x;x_i) =\Delta(x;X)$. In this work, we will use the $i$-th lower envelope $\Delta_i(x;X)$ because we are interested in solving the decoupled partitioning problem (the global lower envelope is relevant to the centralized computation of $\cV(X;\cS)$). Figure~\ref{f:envelope} illustrates the concepts of both the $i$-th lower envelope $\Delta_i$ and the global lower envelope $\Delta$ for a scenario with three agents. To make the illustrations more transparent, we consider an one-dimensional scenario in which the domain $\cS$ is the line segment $[0,1]$ and the set of generators $X$ is the point-set $\{x_0, x_1, x_2\}$ with $0<x_0 < x_1 < x_2 <1$ which are denoted as black crosses in the $x$-axis. In addition, $\delta_i(x;X) = c + \alpha_i(x-x_i)^2$, for $i \in \{0,1,2\}$, with $0 < \alpha_0 < \alpha_1 < \alpha_2$ and $c\geq 0$ (which is in accordance with Assumption~\ref{assumption1}). The graphs of the (generalized) proximity metrics $\delta_i$ and the cells $\cV^i$, for $i\in\{0,1,2\}$ are illustrated with different colors for each agent. The three cells correspond to line segments in $\cS$ whose boundaries are denoted as black squares. We note that $\cV^1$ consists of two disconnected components. The global lower envelope $\Delta$ is illustrated as a dashed curve which corresponds to what an observer sees while looking at the graphs of $\delta_0$, $\delta_1$, and $\delta_2$ from below (from the $x$-axis in Fig.~\ref{f:envelope}). Note that the projection on $\cS$ of the part of the graph of $\Delta$ over which the latter overlaps with the graph of the $i$-th proximity metric $\delta_i$ corresponds to the cell $\cV^i$. The 1st lower envelope $\Delta_1$ (associated with agent $i=1$) is illustrated as a grey dashed-dotted curve. In agreement with Proposition~\ref{prop:globineq}, $\Delta_1\geq 0$ over the two disconnected line segments of $\cS$ that comprise $\cV^1$ and $\Delta_1 < 0$ elsewhere.

\begin{figure}[htb]
\centering
 \psfrag{x}[][][0.8]{$x$}
  \psfrag{y}[][][0.8]{$\delta_i$}
 \psfrag{a}[][][0.8]{$\cV^0$}
  \psfrag{b}[][][0.8]{$\cV^1$}
  \psfrag{c}[][][0.8]{$\cV^2$}
  \psfrag{A}[][][0.8]{$x_0$}
  \psfrag{B}[][][0.8]{$x_1$}
  \psfrag{C}[][][0.8]{$x_2$}
  \psfrag{N}[][][0.8]{$~~~\delta_1$}
  \psfrag{K}[][][0.8]{$~~~\delta_0$}
  \psfrag{M}[][][0.8]{$~~~\delta_2$}
  \psfrag{L}[][][0.8]{$~~~\Delta$}
  \psfrag{D}[][][0.8]{$~~~\Delta_1$}
 \epsfig{file = 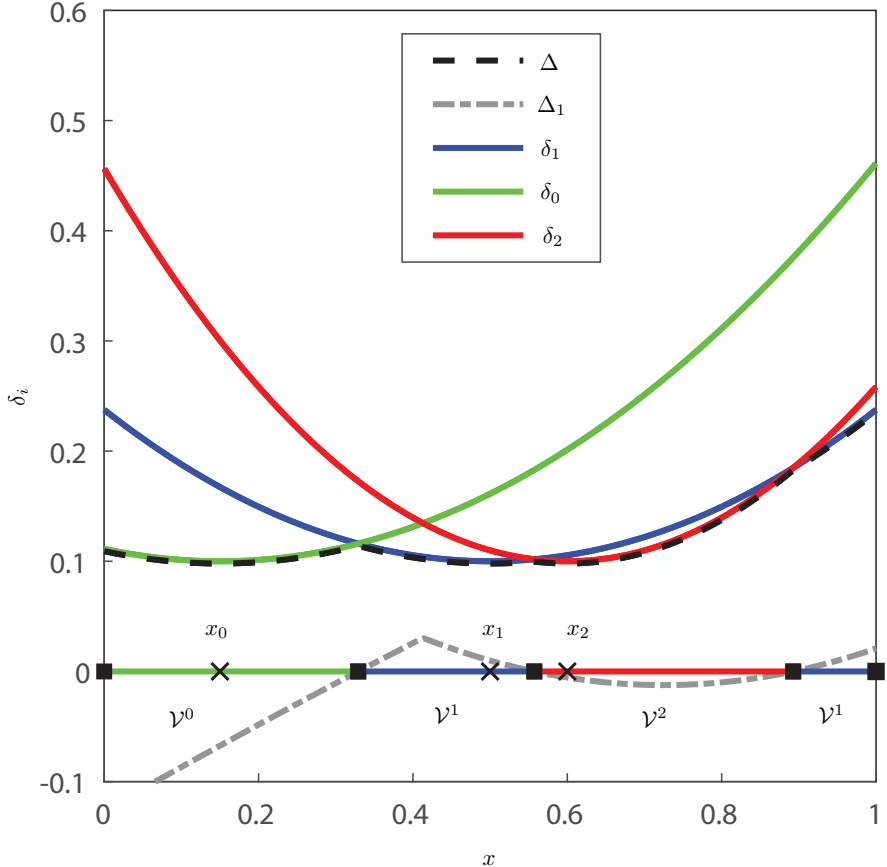,clip=,width=0.68\linewidth}
\caption{\small{Illustration of the global lower envelope function $\Delta(x;X)$ for a network of three agents located in the interval $\cS:=[0,1]$ together with the lower envelope function $\Delta_1(x;X)$ associated with agent $i=1$.}}
\vspace{-2mm} \label{f:envelope}
\end{figure}

It is worth noting that for the computation of $\Delta_i(x;X)$, the
$i$-th agent doesn't need to know neither $x$ nor the set $X$ but
instead the relative position $x-x_i$ and the positions of the other
agents relative to itself (no global reference frame is required).

\begin{proposition}\label{prop:dece}
Let $i\in[0,n]_{\mZ}$. There exists a function $\phi_i:\cS\ominus
\{x_i\} \rightarrow \mathbb{R}$ such that
\begin{equation}
\Delta_i(x;X) = \phi_i(x-x_i;  X \ominus \{x_i\} ),~~\forall
x\in\cS.
\end{equation}
\end{proposition}
\begin{proof}
Indeed, for any $\ell \in [0,n]_{\mZ}\backslash \{i\}$, we have that
\begin{align*}
\delta_{\ell}(x;x_{\ell}) & = (x-x_{\ell})\t \mathbf{P}_{\ell}
(x-x_{\ell}) + \mu_\ell \\
& = (x - x_i + x_i - x_{\ell})\t \mathbf{P}_{\ell} (x - x_i + x_i
-x_{\ell}) + \mu_\ell \\
& = (x - x_i )\t \mathbf{P}_{\ell} (x - x_i) + (x_i - x_{\ell})\t
\mathbf{P}_{\ell} (x_i -x_{\ell}) \\
&~~~ + 2 (x - x_i )\t \mathbf{P}_{\ell} (x_i - x_{\ell}) + \mu_\ell.
\end{align*}
Therefore,
\begin{align*}
\Delta_i(x;X)
& = \min_{\ell \neq i} \big( (x - x_i )\t \mathbf{P}_{\ell} (x - x_i) \\
& ~~~\qquad~~ + ( x_{\ell} - x_i )\t \mathbf{P}_{\ell} ( x_{\ell} - x_i ) \\
& ~~~\qquad~~ - 2 (x - x_i )\t \mathbf{P}_{\ell} (x_{\ell} - x_i) + \mu_\ell \\
& ~~~\qquad~~ - (x - x_i )\t \mathbf{P}_{i} (x - x_i) - \mu_i \big) \\
& =\min_{\ell \neq i} \big( (x - x_i )\t (\mathbf{P}_{\ell}-
\mathbf{P}_{i})(x - x_i)  \\
& ~~~\qquad~~ + (x_{\ell} - x_i)\t \mathbf{P}_{\ell} (x_{\ell} - x_i) \\
& ~~~\qquad~~ - 2 (x - x_i )\t \mathbf{P}_{\ell} (x_{\ell} - x_i) +
\mu_\ell - \mu_i \big).
\end{align*}
Therefore, $\Delta_i(x;X)$ depends on the relative positions $x-x_i$
and $x_i-x_{\ell}$, for $\ell \neq i$. The result follows readily.
\end{proof}

\begin{remark} 
In light of Proposition~\ref{prop:dece}, the computation of the $i$-th lower envelope $\Delta_i$ does not require a global reference frame but it does require, in principle, that all the agents communicate with each other in order to compute the quantity $\min_{\ell \neq i} \delta_{\ell}(x;x_{\ell})$ in a centralized way (all-to-all communication). Later on, however, we will see that the $i$-th agent can characterize $\Delta_i$ by communicating with only a subset of its teammates (the $i$-th agent will find the latter agents by discovering the network topology induced by the HQVP; the latter problem is addressed in Section~\ref{s:netwtopo}), and thus, the computation of $\Delta_i$ can take place in a distributed way.
\end{remark}

\subsection{Parametrization of $\cV^i$ and $\mathrm{bd}(\cV^i)$}

Next, we will show that the cell $\cV^i \in \cV(X;\cS)$ and its
boundary $\mathrm{bd}(\cV^i)$, for $i \in [1, n]_{\mZ}$, admit
convenient parametrizations. These parametrizations will allow us to
propose a systematic way to compute proxies of $\cV^i$ and
$\mathrm{bd}(\cV^i)$ in a finite number of steps. Before we proceed
any further, we introduce some useful notation. In particular, for a
given $i\in [1,n]_{\mathbb{Z}}$ and $\theta \in[0,2\pi[$, we will
denote by $\Gamma_{\theta}$ the ray that starts from $x_i$ and is
parallel to the unit vector $e_\theta =[\cos \theta,~\sin
\theta]\t$, that is, $\Gamma_{\theta} := \{ x \in \mathbb{R}^2:~x=
x_i + \rho e_\theta,~\rho \geq 0 \}$. In addition, we denote as
$\overline{x}_{\theta}$ the point of intersection of
$\Gamma_{\theta}$ with $\mathrm{bd}(E_i \cap \cS)$ where $E_i :=
\mathcal{E}_{\ell_{i,0}}(\mathbf{P}_{i,0}^{-1}
\chi_{i,0}; \mathbf{P}_{i,0}^{-1} )$. 

In view of Proposition~\ref{prop:globineq}, to characterize
$\mathrm{bd}(\cV^i)$ one has to find the roots of $\Delta_i=0$ in
$E_i \cap \cS$ and also check if $\mathrm{bd}(\cV^i)$ contains
boundary points of $\cS$. What we propose to do is to find the roots
of $\Delta_i=0$ incrementally by searching along the ray
$\Gamma_\theta$, or more precisely, the line segment $\Gamma_\theta
\cap (E_i \cap \cS) = [x_i,~\overline{x}_{\theta}]$, for a different
$\theta\in [0,2\pi]$ at each time. For a given $\theta \in
[0,2\pi]$, we will denote as $P_{\theta}^{i}$ the point-set
comprised of the roots of the equation $\Delta_i=0$ in
$[x_i,~\overline{x}_{\theta}[$, that is,
\begin{equation}\label{eq:defnPi}
P_{\theta}^{i} := \{x\in [x_i,~\overline{x}_{\theta}[
:~\Delta_i(x;X) = 0 \}.
\end{equation}
If $P_{\theta}^{i} \neq \varnothing$, then let $M := \mathrm{card}(
P_{\theta}^{i})$ and let us consider the ordered point-set
\[
\mP_{\theta}^{i} =\{ \mfp_m \in [x_i,~\overline{x}_{\theta}]:~ m\in
[0, M+1]_{\mZ} \},
\]
which is comprised of the same points as the set $P_{\theta}^{i}
\cup \{x_i, \overline{x}_{\theta} \}$ with the latter points be
arranged as follows:
\begin{align}\label{p:orderline}
\mfp_0:= x_i,~|x_i - \mfp_1| < \dots < |x_i -
\mfp_M|,~\mfp_{M+1}:=\overline{x}_{\theta}.
\end{align}
The points of $\mP_{\theta}^{i}$ determine a partition $\{I^m:
m\in[1,M+1]_{\mZ}\}$, where $I^m := [\mfp_{m-1} ,~\mfp_m]$, of the
line segment $[x_i,\overline{x}_{\theta}]$.
Next, we provide one of the main results of this section regarding
the characterization of the intersection of the cell $\cV^i$ and its
boundary $\mathrm{bd}(\cV^i)$ with the ray $\Gamma_\theta$.

\begin{proposition}\label{prop:main1}
Let $i\in [1,n]_{\mathbb{Z}}$ and $\theta \in [0, 2\pi]$. Let also
$\{I^m := [\mfp_{m-1} ,~\mfp_m]: m\in[1,M+1]_{\mZ}\}$ be the
partition of $[x_i,~\overline{x}_{\theta}]$ that is induced by the
ordered point-set $\mP_{\theta}^{i}$ whose points are arranged
according to \eqref{p:orderline}. In addition, let $\hat{\mfp}_m$
denote the midpoint of the line segment $I^m$ and let
$\hat{\Delta}^m_i :=\Delta_i(\hat{\mfp}_m;X)$. Further, let us
consider the index-sets
\begin{align*}
\cM^{+} & := \{m \in [1,M+1]_{\mZ}: \hat{\Delta}^m_i>0 \} \\
\cM^{-} & := \{m \in [1,M+1]_{\mZ}: \hat{\Delta}^m_i<0 \}.
\end{align*}
Then,
\begin{equation}\label{eq:thebasiceq}
\cV^i \cap \Gamma_\theta =
\mathfrak{f}_{\theta}(x_i),~~~\mathrm{bd}(\cV^i) \cap \Gamma_\theta
= \mathfrak{g}_{\theta}(x_i),
\end{equation}
where the set-valued maps $\mathfrak{f}_\theta(\cdot) : X
\rightrightarrows \wp([x_i,~\overline{x}_{\theta}])$ and
$\mathfrak{g}_\theta(\cdot) : X \rightrightarrows
\wp([x_i,~\overline{x}_{\theta}])$
are defined as follows:

\noindent $\mathrm{i})$ If $P_{\theta}^{i}=\varnothing$, then
\begin{align*}
\mathfrak{f}_{\theta}(x_i) :=
[x_i,\overline{x}_{\theta}],~~~\mathfrak{g}_{\theta}(x_i) := \{
\overline{x}_{\theta} \}.
\end{align*}

\noindent $\mathrm{ii})$ If $P_{\theta}^{i} \neq \varnothing$, then
\begin{equation*}
\mathfrak{f}_{\theta}(x_i) := \bigcup_{m \in \cM_{\mff}}
I^{m},~~~\mathfrak{g}_{\theta}(x_i) := \{\mfp_m: m \in \cM_{\mfg}
\},
\end{equation*}
where $\cM_{\mff} = \cM^{+}$ 
and $\cM_{\mfg} := \cM_{\mfg}^+ \cup \cM_{\mfg}^{-}$.
In particular, the index-set $\cM_{\mfg}^{+}$ is comprised of all $m
\in \cM^{+} \cap [1,M]_{\mZ}$ such that $m+1 \in \cM^{-}$ plus the
index $M+1$ if $M+1 \in \cM^{+}$. Finally, the index-set
$\cM_{\mfg}^{-}$ is comprised of all $m \in \cM^{-} \cap
[1,M]_{\mZ}$ such that $m+1 \in \cM^{+}$.
\end{proposition}
\begin{proof}
First, we consider the case when $P_{\theta}^{i}=\varnothing$, that
is, $\Delta_i$ has no roots in $[x_i, \overline{x}_{\theta}[$. In
view of Assumption~\ref{assumption1}, we have
\[
\Delta_i(x_i;X) = \min_{\ell \neq i} \delta_{\ell}(x_i;x_{\ell}) -
\mu_i > 0.
\]
By continuity, we conclude that in this case $\Delta_i(x;X)
> 0$, for all $x\in [x_i, \overline{x}_{\theta}[$, which implies that
$\mathfrak{f}_{\theta}(x_i) = \cV^i \cap \Gamma_\theta = [x_i,
\overline{x}_{\theta}]$ and $\mathfrak{g}_{\theta}(x_i) =
\mathrm{bd}(\cV^i) \cap \Gamma_\theta = \{ \overline{x}_{\theta}\}$.

Next, we consider the case when $P_{\theta}^{i} \neq \varnothing$.
By definition, $\hat{\Delta}^m_i>0$ for all $m \in \cM^{+}$. By
continuity of $\Delta_i$, we have that $\Delta_i(x;X) \geq 0$ for
all $x\in I_m = [\mfp_{m-1},~\mfp_m]$ and for all $m \in \cM^{+}$.
Therefore, in view of equation \eqref{eq:globineq}, we have
$\bigcup_{m \in \cM_{\mff}}I_m = \cV^i \cap \Gamma_{\theta} =
\mathfrak{f}_{\theta}(x_i)$ with $\cM_{\mff} = \cM^{+}$. 
Now, a point $\mfp_m$ with $m\in [1,M]_{\mZ}$ belongs to
$\mathrm{bd}(\cV^i) \cap \Gamma_\theta = \mathfrak{g}_{\theta}(x_i)$
if and only if as one transverses $\Gamma_\theta$ (with direction
from $x_i$ towards $\overline{x}_{\theta}$), one of the following
two events takes place: 1) $\Delta_i$, which is negative ``before''
$\mfp_{m}$, becomes positive ``after'' $\mfp_{m}$ (in which case $m
\in \cM_{\mfg}^{+}$) or 2) $\Delta_i$, which is positive ``before''
$\mfp_{m}$, becomes negative ``after'' $\mfp_{m}$ (in which case $m
\in \cM_{\mfg}^{-}$). Finally, if $\Delta_i(\hat{\mfp}_{M+1};X) >
0$, then $\mfp_{M+1} \in \mathfrak{g}_{\theta}(x_i)$ and $M+1\in
\cM_{\mfg}^{+}$. This completes the proof.
\end{proof}

\noindent \textit{Example:} To better understand the implications of
Proposition~\ref{prop:main1} as well as the meaning of each
index-set introduced therein, let us consider the example
illustrated in Figure~\ref{f:invpend}. We have that
$\mathfrak{P}_\theta^i =\{ \mfp_m:~m\in[0,5]_{\mZ} \}$ where
$\mfp_0\equiv x_i$ and $\mfp_5 \equiv \overline{x}_\theta$ and its
induced partition is $\{ I_m = [\mfp_{m-1}, \mfp_m]:~m\in[1,5]_{\mZ}
\}$. The sign of $\Delta_i$ at the mid-points of the segments $I_1,
I_2, I_4, I_5$, which are enclosed by dashed blue ellipses in the
figure, is positive and thus $\cM^{+} =\cM_{\mff} = \{ 1, 2, 4, 5
\}$ whereas $\cM^{-} =\{ 3 \}$. We conclude that
$\mathfrak{f}_{\theta}(x_i) = [x_i,\mfp_2] \cup [\mfp_3,
\overline{x}_{\theta}]$. Furthermore, as one transverses
$\Gamma_\theta$ (from $x_i$ towards $\overline{x}_{\theta}$)
$\Delta_i$ changes sign from positive to negative at $\mfp_2$ and in
addition, $\Delta_i>0$ at the mid-point of  $I_5$; thus,
$\cM^{+}_\mfg = \{2, 5\}$. Also, $\Delta_i$ changes sign from
negative to positive at $\mfp_3$, and thus $\cM^{-}_{\mfg} =\{3\}$.
Hence, $\cM_{\mfg} = \cM_{\mfg}^{+} \cup \cM_{\mfg}^{-} = \{2 ,3, 5
\}$. We conclude that $\mathfrak{g}_{\theta}(x_i) = \{\mfp_2,
\mfp_3, \mfp_5 \}$. The points from $\mathfrak{P}^i_{\theta}$ that
form $\mathfrak{g}_{\theta}(x_i)$ are encircled by blue circles in
Figure~\ref{f:invpend}.

\begin{remark}
A careful interpretation of the results presented in Proposition~\ref{prop:main1} reveals that under some mild and intuitive modifications, one can characterize the cell $\cV^i$ and its boundary $\mathrm{bd}(\cV^i)$ even for the more general case when Assumption~1 may not hold true. For instance, in the previous example, the sign of $\Delta_i$ in the segment $]x_i,\mfp_1[$ will not necessarily be positive (it is always positive if Assumption~1 holds true) and, instead, it will be equal to the sign of $\Delta_i$ at any interior point in that segment. For the sake of the argument, let us take the latter sign to be negative. Then, assuming that the signs of $\Delta_i$ in all the other segments remain the same as in Fig.~\ref{f:invpend}, it follows that $\mathfrak{g}_{\theta}(x_i) = \{\mfp_1, \mfp_2,
\mfp_3, \mfp_5 \}$ and $\cM^{+} =\cM_{\mff} = \{2, 4, 5
\}$.
\end{remark}

\begin{proposition}\label{prop:main2}
Let us consider a family of rays $\{ \Gamma_\theta:~\theta\in [0, 2
\pi] \}$, where the ray $\Gamma_\theta$ emanates from $x_i$ and is
parallel to the unit vector $e_\theta:=[\cos \theta,~\sin
\theta]\t$. Then,
\begin{align}
\cV^i & = \bigcup_{\theta \in [0,2\pi]}
\mathfrak{f}_{\theta}(x_i), \qquad 
\mathrm{bd}(\cV^i) = \bigcup_{\theta \in [0,2\pi]}
\mathfrak{g}_{\theta}(x_i), \label{eq:bdparameta2}
\end{align}
where the set-valued maps $\mathfrak{f}_{\theta}(\cdot)$ and
$\mathfrak{g}_{\theta}(\cdot)$ are defined as in
Proposition~\ref{prop:main1} for each $\theta\in[0,2\pi]$.
\end{proposition}
\begin{proof}
We have that
\begin{align*}
\bigcup_{\theta \in [0,2\pi]} \mathfrak{f}_{\theta}(x_i) =
\bigcup_{\theta \in [0,2\pi]} (\cV^i \bigcap \Gamma_{\theta})= \cV^i
\bigcap (\bigcup_{\theta \in [0,2\pi]} \Gamma_{\theta})  = \cV^i,
\end{align*}
where in the first equality, we used \eqref{eq:thebasiceq} and in
the last one, we used that $\bigcup_{\theta \in [0,2\pi]}
\Gamma_{\theta} = \mathbb{R}^2$. Thus, we have proved that the first
equation in \eqref{eq:bdparameta2} holds true. The proof for the
second one follows similarly. 
\end{proof}

\begin{figure}[htb]
\centering
 \psfrag{a}[][][0.8]{$\mfp_1$}
  \psfrag{b}[][][0.8]{$~~\mfp_2$}
    \psfrag{c}[][][0.8]{$~~\mfp_3$}
    \psfrag{y}[][][0.8]{$\mfp_4$}
    \psfrag{X}[][][0.8]{$~~~\qquad~~\mfp_5 \equiv \overline{x}_{\theta}$}
    \psfrag{x}[][][0.8]{$\mfp_0 \equiv x_i~~~~$}
    \psfrag{A}[][][0.8]{$\Delta_i > 0~~~$}
    \psfrag{d}[][][0.8]{$~~~~\Delta_i > 0~~~$}
    \psfrag{f}[][][0.8]{$\Delta_i < 0$}
    \psfrag{D}[][][0.8]{$~~~~\Delta_i > 0$}
    \psfrag{F}[][][0.8]{$~~~~~~~\Delta_i > 0$}
    \psfrag{G}[][][0.8]{$~~~\Gamma_{\theta}$}
    \psfrag{E}[][][0.8]{$E_{i}$}
 \epsfig{file = 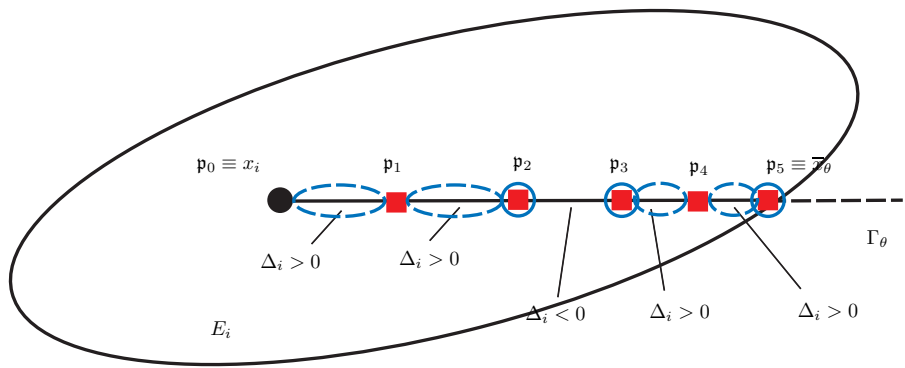,clip=,width=0.69\linewidth}
\caption{\small{Illustrative example on the characterization of
$\cV^i \cap \Gamma_{\theta}$ and $\mathrm{bd}(\cV^i) \cap
\Gamma_{\theta}$ based on Proposition \ref{prop:main1}.}}
\vspace{-2mm} \label{f:invpend}
\end{figure}

\section{A systematic approach for the computation of a finite
approximation of $\mathrm{bd}(\cV^i)$ and $\cV^i$}\label{s:mainalgo}

\subsection{Efficient computation of the roots of the equation $\Delta_i=0$}

In this section, we will leverage Propositions~\ref{prop:main1} and
\ref{prop:main2} to develop a systematic procedure to characterize
the boundary points of $\cV^i$ that lie on a given ray
$\Gamma_{\theta}$ after a finite number of steps. To this aim, let
$\overline{\rho}_{\theta}>0$ denote the length of
$[x_i,~\overline{x}_{\theta}]$, that is, $\overline{\rho}_{\theta}
:= | \overline{x}_{\theta} - x_i|$. Recall that
$\overline{x}_{\theta}$ corresponds to the intersection of
$\Gamma_{\theta}$ with $\mathrm{bd}(E_i \cap \cS)$.
In addition, let
\begin{align}\label{eq:isoeq}
\mathcal{R}^{i,j}_{\theta} & :=\{\rho \in
[0,~\overline{\rho}_{\theta}[:~\delta_{i}(x_i + \rho \e_\theta; x_i)
= \nonumber \\&~~~~\qquad~~~~~~~~~~~~~~~~ \delta_{j}(x_i + \rho
\e_\theta; x_j)\},
\end{align}
for $j \in [0,n]_{\mathbb{Z}}\backslash\{i\}$. Equivalently,
$\mathcal{R}^{i,j}_{\theta}$ consists of all $\rho \in
[0,~\overline{\rho}_{\theta}]$ that satisfy the following equation:
\begin{equation}\label{eq:quadreq}
\alpha \rho^2 + \beta \rho + \gamma = 0,
\end{equation}
where
\begin{subequations}
\begin{align}
\alpha & := |\bm{\Pi}_i^{1/2}
\e_\theta |^2 - |\bm{\Pi}_j^{1/2} \e_\theta |^2,\\
\beta & := 2 (x_j - x_i)\t \bm{\Pi}_j \e_\theta, \\
\gamma & := |\bm{\Pi}_j^{1/2} ( x_i - x_j )|^2 + \mu_i - \mu_j.
\end{align}
\end{subequations}

Note that if $\rho \in \mathcal{R}^{i,j}_{\theta}$, then the point
$p := x_i + \rho \e_\theta$ will belong to $\Gamma_{\theta} \cap
\cB_{i,j}$. Let $\mathcal{P}^{i,j}_{\theta}:=\{p \in [x_i,
\overline{x}_{\theta}[:~p= x_i + \rho \e_\theta,~\rho \in
\mathcal{R}^{i,j}_{\theta} \}$. Note that there is an (obvious)
one-to-one correspondence between the point-sets
$\mathcal{P}^{i,j}_{\theta}$ and $\mathcal{R}^{i,j}_{\theta}$, which
may both be empty for some $j \neq i$. Now, let
\begin{equation}\label{eq:RPitheta}
\mathcal{R}^i_{\theta} := \cup_{j\neq i }
\mathcal{R}^{i,j}_{\theta},~~~~\cP^i_{\theta} := \cup_{j\neq i }
\mathcal{P}^{i,j}_{\theta}.
\end{equation}
Note that a point $p \in
\cP^i_{\theta}\backslash\{x_i,\overline{x}_{\theta}\}$ is
necessarily equidistant from the $i$-th agent and at least a
different agent from the same extended network. This naturally leads
us to the following proposition.

\begin{proposition}\label{prop:RPbasic}
Let $P^i_{\theta}$ be the point-set which is defined as in
\eqref{eq:defnPi}. 
Then,
$\cP^i_{\theta} \supseteq P^i_{\theta}$ and thus,
\begin{align}\label{eq:basincl0}
P^i_{\theta} & =\{ x \in \cP^i_{\theta}: \Delta_i(x;X)=0\}.
\end{align}
\end{proposition}
\begin{proof}
The proof follows readily from the definitions of $\cP^i_{\theta}$
and $P^i_{\theta}$.
\end{proof}

Proposition \ref{prop:RPbasic} implies that for the characterization
of the set $P^i_{\theta}$ that consists of all the roots of
$\Delta_i=0$ in $[x_i,\overline{x}_{\theta}[$, one has to evaluate
the function $\Delta_i$ at the points of the finite point-set
$\cP^i_{\theta}$, which is a superset of the unknown set
$P^i_{\theta}$. In particular, $P^i_{\theta}$ is comprised of all
those points of $\cP^i_{\theta}$ at which $\Delta_i$ vanishes and
only them. 

\subsection{Line search algorithm for the computation of
$\mathfrak{f}_{\theta}(x_i)$ and $\mathfrak{g}_{\theta}(x_i)$}

Next, we present an algorithm that computes
$\mathfrak{f}_{\theta}(x_i)$ and $\mathfrak{g}_{\theta}(x_i)$ for a
given $\theta \in [0,2\pi[$ based on the previous discussion and analysis. The main steps of the
proposed algorithmic process can be found in
Algorithm~\ref{DVicell}. In particular, the first step is to compute
the point-set $\cP_\theta^i$ (line 5). If $\cP^i_{\theta} =
\varnothing$, then we set $\mathfrak{f}_{\theta}(x_i)$ and
$\mathfrak{g}_{\theta}(x_i)$ to be equal to, respectively, $[x_i,
\overline{x}_\theta]$ and $\{ \overline{x}_{\theta} \}$ and the
process is complete (lines 6-7). If $\cP^i_{\theta} \neq
\varnothing$, we characterize all of the points in $\cP^i_{\theta}$
that correspond to the roots of the equation $\Delta_i = 0$ in
$[x_i,\overline{x}_{\theta}[$ to form the point-set $P_{\theta}^i$
in accordance with Proposition~\ref{prop:RPbasic} (line 8). Next, we
apply a permutation to the point-set $P_{\theta}^i \cup \{x_i,
\overline{x}_{\theta} \}$ to obtain the point-set $\mP_{\theta}^i =
\{\mfp_m:~m\in[1,M+1]_{\mZ} \}$ whose points are ordered in
increasing distance from $x_i$ as in \eqref{p:orderline} (lines
9-10). Next, we start an iterative process for the characterization
of the index sets $\cM_{\mff}$ and $\cM_{\mfg}$, with
$\cM_{\mff}=\cM^{+}$ and $\cM_{\mfg}=\cM^+_{\mfg} \cup \cM^-_{\mfg}$
(lines 11-25), where the index sets $\cM^{+}$, $\cM^+_{\mfg}$ and
$\cM^-_{\mfg}$ are defined as in Proposition~\ref{prop:main1}.
Finally, we set $\mathfrak{f}_{\theta}(x_i) := \cup_{m \in
\cM_{\mff}}[\mfp_{m-1}, \mfp_m]$ and $\mathfrak{g}_{\theta}(x_i) :=
\{\mfp_m \in \mP_{\theta}^i:~m\in \cM_{\mfg} \}$ (lines 26-27).

Note that after the computation of $\mathfrak{g}_{\theta}(x_i)$,
then, in view of Proposition~\ref{prop:main1}, one can compute an
approximation of $\mathrm{bd}(\cV^i)$ by computing
$\mathfrak{g}_{\theta}(x_i)$ for all $\theta \in \Theta$, where
$\Theta$ is a finite point-set whose points define a partition of
$[0,2\pi]$.

\begin{algorithm}
\caption{Computation of point-sets $\mathfrak{f}_{\theta}(x_i) =
\cV^i \cap \Gamma_{\theta}$ and $\mathfrak{g}_{\theta}(x_i) =
\mathrm{bd}(\cV^i) \cap \Gamma_{\theta}$}\label{DVicell}
\begin{algorithmic}[1]
\Procedure{Cell computation}{} 
\BState \textit{Input data}: $X$, $\{ (\mathbf{P}_{\ell},
\mu_{\ell}):~\ell\in[0,n]_{\mZ}\}$ 
\BState \textit{Input variables}: $i$, $\theta$ 
\BState \textit{Output variables}:
$\mathfrak{f}_{\theta}(x_i),\mathfrak{g}_{\theta}(x_i)$ \State Find
$\cP^i_{\theta}$ \If {$\cP^i_{\theta} = \varnothing$} \State
$\mathfrak{f}_{\theta}(x_i)\gets [x_i, \overline{x}_{\theta}]$,
$\mathfrak{g}_{\theta}(x_i)\gets \{ \overline{x}_{\theta} \}$
\Return \EndIf \State Extract point-set $P^i_{\theta}$ from
$\cP^i_{\theta}$ based on Proposition~\ref{prop:RPbasic} \State
$\mathfrak{P}_{\theta}^i \gets
P^i_{\theta}\cup\{x_i,\overline{x}_{\theta}\}$. \State Re-arrange
points in $\mathfrak{P}_{\theta}^i = \{\mfp_m:~m\in[0,M+1]_{\mZ}\}$
based on increasing distance from $x_i$ according to
\eqref{p:orderline}
\State $\cM_{\mff} \gets \varnothing,~\cM_{\mfg}^{+} \gets
\varnothing,~\cM_{\mfg}^{-} \gets \varnothing$ \For{ $m=1:M+1$}
\State $\hat{x} \gets \tfrac{1}{2}(\mfp_{m-1} + \mfp_{m})$
\Comment{$\hat{x}$: midpoint of $I^{m}$} \State $\hat{\Delta} \gets
\Delta_i(\hat{x};X)$ \If{$\hat{\Delta}>0$} $\cM_{\mff} \gets
\cM_{\mff} \cup \{m\}$ \EndIf
 \If{ $m<M+1$} \State $\hat{x}' \gets \tfrac{1}{2}(\mfp_{m} +
\mfp_{m+1})$ \Comment{$\hat{x}'$: midpoint of $I^{m+1}$} \State
$\hat{\Delta}' \gets \Delta_i(\hat{x}';X)$ \If{ $\hat{\Delta}>0$ and
$\hat{\Delta}'<0$ } \State $\cM_{\mfg}^{+} \gets \cM_{\mfg}^{+} \cup
\{m\}$ \EndIf \If{ $\hat{\Delta} <0$ and $\hat{\Delta}'>0$ } \State
$\cM_{\mfg}^{-} \gets \cM_{\mfg}^{-} \cup \{m\}$ \EndIf \EndIf \If{
$\hat{\Delta}>0$ and $m = M+1$} \State $\cM^{+}_{\mfg} \gets
\cM^{+}_{\mfg} \cup \{ m\}$ \EndIf \EndFor
\State $\cM_{\mfg} \gets \cM^{+}_{\mfg} \cup \cM^{-}_{\mfg}$ \State
$\mathfrak{f}_{\theta}(x_i) \gets \{[\mfp_{m-1}, \mfp_m]:~m \in
\cM_{\mff} \}$ \State $\mathfrak{g}_{\theta}(x_i) \gets \{\mfp_m:~m
\in \cM_{\mfg} \}$
\EndProcedure
\end{algorithmic}
\end{algorithm}

\section{Discovery of Network Topology Induced by HQVP}\label{s:netwtopo}

In order to solve Problem~\ref{problemDP} in a distributed way, it
is necessary that the $i$-th agent can discover a superset of its
neighbors in the topology of HQVP before even computing its own
cell. Next, we characterize an upper bound on the distance of the
$i$-th agent, measured in terms of $\delta_i$, from the points in
its own cell.

\begin{proposition}\label{prop:boundD}
Let $E_i :=\mathcal{E}_{\ell_{i,0}}(\mathbf{P}_{i,0}^{-1}
\chi_{i,0}; \mathbf{P}_{i,0}^{-1} )$. Then,
\begin{equation}\label{eq:ineqbd}
\delta_i(x;x_i) \leq \overline{\delta}_i,~~~\forall~x\in \cV^i,
\end{equation}
where $\overline{\delta}_i :=\max \{\delta_i(x;x_i):~x\in
\mathrm{bd}(E_i \cap \cS) \}$.
\end{proposition}
\begin{proof}
Because $\delta_i(x;x_i)$ is a convex quadratic function, we
conclude that its restriction over the convex and compact set $E_i
\cap \cS$ attains its maximum value in the latter set and in
addition, at least one of its maximizers belongs to the boundary
$\mathrm{bd}(E_i \cap
\cS)$ of the same set. Consequently, 
\begin{align*}
\overline{\delta}_i & = \max \{\delta_i(x;x_i):~x\in E_i \cap \cS \}
\\ &= \max \{\delta_i(x;x_i):~x\in \mathrm{bd}(E_i \cap \cS) \}.
\end{align*}
Inequality~\eqref{eq:ineqbd} follows from the set
inclusion~\eqref{eq:mainsetincl}.
\end{proof}

\begin{proposition}\label{prop:neigh}
Let us consider the index-set $\widetilde{\cN}_i$ which is defined
as follows:
\begin{equation*}
\widetilde{\cN}_i := \{\ell \in [0,n]_{\mZ}\backslash \{i\}:
\delta_{\ell}(x;x_{\ell}) \leq \overline{\delta}_i,~\forall x\in
\mathrm{bd}(E_i \cap \cS) \}, 
\end{equation*}
where $\overline{\delta}_i :=\max \{\delta_i(x;x_i):~x\in
\mathrm{bd}(E_i \cap \cS) \}$. Then, the set inclusion
$\widetilde{\cN}_i \supseteq \cN_i$ holds true.
\end{proposition}
\begin{proof}
In view of Proposition~\ref{prop:basicprop}, all points in
$\mathrm{bd}(\cV^i) \backslash \mathrm{bd}(\cS)$ are equidistant
from at least one different agent from the same network, that is,
for any point $x \in \mathrm{bd}(\cV^i) \backslash
\mathrm{bd}(\cS)$, there exists $j_x \in [0,n]_{\mZ} \backslash\{ i
\}$ (the index $j_x$ depends on $x$) such that $\delta_i(x;x_i) =
\delta_{j_x}(x;x_{j_x})$. Thus, in view of
Definition~\ref{def:Topology}, $j_x \in \cN_i$. Now let $\ell \neq i$ and let us assume that $\ell \in
\widetilde{\cN}_i^c$, where $\widetilde{\cN}_i^c := \{ \ell \in
[0,n]_{\mZ} \backslash \{i\}: \ell \notin \widetilde{\cN}_i\}$.
Then, $\delta_{\ell}(x;x_{\ell})
> \overline{\delta}_i$, $\forall x \in \mathrm{bd}(E_i \cap \cS)$. But, in view of
Proposition~\ref{prop:boundD}, $\delta_i(x;x_i) \leq
\overline{\delta}_i,~~\forall x\in \cV^i \supsetneq
\mathrm{bd}(\cV^i)$; consequently, there is no point
$x\in\mathrm{bd}(\cV^i)$ such that $\delta_i(x;x_i) =
\delta_{\ell}(x;x_{\ell})$. Thus, $\ell \in \cN_i^c$ where $\cN_i^c
:= \{ \ell \in [0,n]_{\mZ} \backslash \{i\}: \ell \notin \cN_i\}$,
which implies that $\widetilde{\cN}_i^c \subseteq \cN_i^c$. We
conclude that $\widetilde{\cN}_i \supseteq \cN_i$ and the proof is
complete.
\end{proof}

Next, we will leverage Proposition~\ref{prop:neigh} to show that the
$i$-th agent can find a subset of the spatial domain $\cS$ that will
necessarily contain its neighbors without having computed $\cV^i$.

\begin{proposition}\label{prop:neincl}
Let $i\in [1,n]_{\mZ}$ and let $\mathcal{A}_i$ denote the compact
set enclosed by the closed curve $\cC_i: [0,2\pi] \rightarrow
\mathbb{R}^2$ with
\begin{align}\label{eq:xphi}
\cC_i(\phi) & := \mathbf{P}_{i,0}^{-1} \chi_{i,0} +
\sqrt{\ell_{i,0}}\mathbf{P}_{i,0}^{-1/2} e_\phi \nonumber \\
& ~~~~ + (\sqrt{r} / \| \mathbf{P}_{0}^{-1/2} \mathbf{P}_{i,0}^{1/2}
e_\phi \|) \mathbf{P}_{0}^{-1}\mathbf{P}_{i,0}^{1/2} e_\phi,
\end{align}
where $e_{\phi}:=[\cos\phi,~\sin \phi]^{\mathrm{T}}$.
Then, all the neighbors of the $i$-th agent lie necessarily in
$\cA_i$, that is,
\begin{equation}\label{eq:neincl}
x_\ell \in \mathcal{A}_i \cap X,~~~\forall~\ell \in \cN_i.
\end{equation}
\end{proposition}
\begin{proof}
Let $\mathsf{w} \in \mathrm{bd}(E_i)$, where $E_i :=
\mathcal{E}_{\ell_{i,0}}(\mathbf{P}_{i,0}^{-1} \chi_{i,0};
\mathbf{P}_{i,0}^{-1} )$, and let us consider a point $\mz$ such
that the intersection of the ellipsoid $\mathcal{E}_{r_i}(\mz;
\mathbf{P}_{0}^{-1})$, where $r_i := \overline{\delta}_i - \mu_0$
(note that $r_i>0$ in view of Assumption~\ref{assumption2}), with
$E_i$ corresponds to the singleton $\{ \mathsf{w} \}$, that is,
\begin{align*}
\{ \mathsf{w} \} & = E_i \cap \mathcal{E}_{r_i}(\mz;
\mathbf{P}_{0}^{-1} ) = \mathrm{bd}(E_i) \cap
\mathrm{bd}(\mathcal{E}_{r_i}(\mz; \mathbf{P}_{0}^{-1} )).
\end{align*}
Because $\mathsf{w} \in \mathrm{bd}(E_i ) \cap
\mathrm{bd}(\mathcal{E}_{r_i}(\mz; \mathbf{P}_{0}^{-1} ))$,
\begin{align*}
0& = \|\mathbf{P}_{i,0}^{1/2} (\mathsf{w} - \mathbf{P}_{i,0}^{-1}
\chi_{i,0})\| - \sqrt{\ell_{i,0}} 
 = \| \mathbf{P}_{0}^{1/2} (\mathsf{w}-\mathsf{z})\| - \sqrt{r_i},
\end{align*}
which implies that there exist $\phi, \varphi \in [0,2\pi[$ such
that
\[
\mathsf{w} = \mathbf{P}_{i,0}^{-1} \chi_{i,0} + \sqrt{\ell_{i,0}}
\mathbf{P}_{i,0}^{-1/2} e_\phi = \mathsf{z} + \sqrt{r_i}
\mathbf{P}_{0}^{-1/2} e_{\varphi},
\]
where $e_{\phi}:=[\cos\phi,~\sin \phi]^{\mathrm{T}}$ and
$e_{\varphi}:=[\cos\varphi,~\sin \varphi]^{\mathrm{T}}$. Thus,
\[
\mathsf{z} = \mathbf{P}_{i,0}^{-1} \chi_{i,0} + \sqrt{\ell_{i,0}}
\mathbf{P}_{i,0}^{-1/2} e_\phi - \sqrt{r_i} \mathbf{P}_{0}^{-1/2}
e_{\varphi}.
\]
The normal vectors of the ellipsoids $E_i$ and
$\mathcal{E}_{r_i}(\mz; \mathbf{P}_{0}^{-1} )$ at point $\mathsf{w}$
(contact point) are anti-parallel, that is, there exists $\lambda >
0$ such that
\begin{align*}
&\tfrac{\partial}{\partial x}\big((x-
\mathbf{P}_{i,0}^{-1}\chi_{i,0} )\t \mathbf{P}_{i,0} (x-
\mathbf{P}_{i,0}^{-1}\chi_{i,0} ) -
\ell_{i,0} \big)\big|_{x=\mathsf{w}} \\
&~~\qquad~~~\qquad~~= - \lambda \tfrac{\partial}{\partial x}\big((x-
\mz )\t \mathbf{P}_{0} (x- \mz ) -r_i \big)\big|_{x=\mathsf{w}},
\end{align*}
from which it can be shown (see, for instance, Lemma~5 in
\cite{p:colavoidell17}) that
\[
e_{\varphi} = - (1/\| \mathbf{P}_{0}^{-1/2} \mathbf{P}_{i,0}^{1/2}
e_\phi \|) \mathbf{P}_{0}^{-1/2} \mathbf{P}_{i,0}^{1/2} e_\phi
\]
and thus, we conclude that $\mathsf{z} = \cC_i(\phi)$ where
$\cC_i(\phi)$ is defined in \eqref{eq:xphi}.

Now, let $\cA_i$ be the compact set enclosed by the closed curve
$\cC_i$. We will show that all the neighbors of the $i$-th agent are
located in $\mathcal{A}_i$, that is, $\mathcal{A}_i \supsetneq
\{x_{k} \in X:~k\in \cN_i\}$. In view of
Proposition~\ref{pr:setcontain}, the set inclusion
$\mathcal{E}_{r_i}(\mathsf{z}; \mathbf{P}_{0}^{-1} ) \supsetneq
\mathcal{E}_{r_i}(\mz; \mathbf{P}_{\ell}^{-1} )$ holds true for all
$\mz \in \cC_i$ and for all $\ell \neq i$. Now, for a given $\mz \in
\cC_i$, we have that
\[
\delta_0(x; \mz) = (x - \mz)\t \mathbf{P}_{0} (x - \mz) + \mu_0 =
\overline{\delta}_i,
\]
for all $x \in \mathrm{bd}(\mathcal{E}_{r_i}(\mz;
\mathbf{P}_{0}^{-1} ))$ whereas
\[\delta_\ell(y;\mz) = (y - \mz)\t \mathbf{P}_{\ell} (y-
\mz) + \mu_\ell = \overline{\delta}_i + \mu_{\ell}-\mu_0,
\]
for all $y \in \mathrm{bd}(\mathcal{E}_{r_i}(\mathsf{z};
\mathbf{P}_{\ell}^{-1} ))$. Because, $\mu_{\ell}-\mu_0\geq 0$, we
conclude that $\max \{ \delta_\ell(y;\mz): y \in
\mathcal{E}_{r_i}(\mz; \mathbf{P}_{\ell}^{-1} ) \} \geq \max
\{\delta_0(x;\mz): x \in \mathcal{E}_{r_i}(\mz; \mathbf{P}_{0}^{-1}
)\}$ which together with the set inclusion
$\mathcal{E}_{r_i}(\mathsf{z}; \mathbf{P}_{0}^{-1} ) \supsetneq
\mathcal{E}_{r_i}(\mathsf{z}; \mathbf{P}_{\ell}^{-1} )$ imply that
$\delta_{\ell}(y;\mz)> \overline{\delta}_i$ for all $y \in E_i
\supseteq E_i \cap \cS \supseteq \cV^i$ (the last set inclusion
follows from Proposition~\ref{prop:main}). Therefore,
\begin{align}
\delta_\ell(x;\mz)> \overline{\delta}_i \geq \max\{
\delta_\ell(y;x_i):~y \in \cV^i \},~~~\forall x \in \cV^i.
\end{align}
Therefore, there is no point $x \in \mathrm{\bd}(\cV^i)$ such that
$\delta_\ell(x;\mz) = \delta_\ell(x;x_i)$ for any $\mz \in \cC_i$.
Thus, in view of Proposition \ref{prop:basicprop}, it follows that
$\ell \notin \cN_i$ and the proof is complete.
\end{proof}

Proposition~\ref{prop:neincl} implies that the neighbors of the
$i$-th agent are necessarily confined in the subset $\cA_i \subseteq
\cS$ which is known to this agent before computing its cell $\cV^i$.
In practice, the $i$-th agent can communicate and exchange
information directly with its neighbors (e.g., by means of
point-to-point communication) provided that its communication radius
$\eta_i>0$ is sufficiently large such that its communication region
$\cB_{\eta_i}(x_i) \supseteq \cA_i$.

\begin{proposition}\label{prop:lowerbdCR}
The neighbors of the $i$-th agent are necessarily located in the
communication region $\cB_{\eta_i}(x_i)$ of the $i$-th agent, that
is,
\begin{align}\label{eq:setcomm}
\cB_{\eta_i}(x_i) \supseteq \{x_k \in X:~k\in \cN_i\},~~~~\forall
\eta_i\geq \underline{\eta_i}
\end{align}
where $\underline{\eta_i} := \max_{\phi \in [0,2\pi] } \|
\cC_i(\phi) -x_i\|$, 
with $\cC_i(\phi)$ defined as in \eqref{eq:xphi}.
\end{proposition}
\begin{proof}
By the definition of $\underline{\eta_i}$, we have that
\[
\cB_{\eta_i}(x_i) \supsetneq \{ \cC_i(\phi):~\phi\in[0,2\pi]\} =
\mathrm{bd}(\cA_i),
\]
and thus $\cB_{\eta_i}(x_i) \supseteq \cA_i$,
for all $\eta_i\geq \underline{\eta_i}$. Because $\cA_i$ contains all the neighbors of the $i$-th agent in
view of Proposition~\ref{prop:neincl}, then so does the closed ball
$\cB_{\eta_i}(x_i)$, for any $\eta_i\geq \underline{\eta_i}$. Thus,
the set inclusion \eqref{eq:setcomm} holds true.
\end{proof}

\begin{proposition}\label{prop:distrVP}
Let $\cI\subseteq \cN_i \subseteq \widetilde{\cN}_i \subseteq [0,n]_{\mathbb{Z}}$, where the index sets $\cN_i$ and $\widetilde{\cN}_i$ are defined as in \eqref{eq:Ni} and Proposition~\ref{prop:neigh}, respectively, and let $\Delta_i^{\cI}(x) := \min_{\ell \in \cI \backslash \{i\} } \delta_{\ell}(x;x_{\ell}) - \delta_i(x;x_i)$. Then 
\begin{equation}\label{eq:DeltaiI}
\Delta_i^{\cI}(x;x_i)= \Delta_i(x;x_i)=0,~~~\forall x \in \mathrm{bd}(\cV^i) \backslash \mathrm{bd}(\cS),
\end{equation}
where $\Delta_i(x;x_i)$ is defined as in \eqref{eq:Deltaeq}.
\end{proposition}
\begin{proof}
By definition, $\Delta_i^{\cI}(x) \geq \Delta_i(x;X)$, for all $x\in \cS$, given that the $\min$ operator in the definition of $\Delta_i^{\cI}$ is applied over an index set which is a subset of the one that appears in the definition of $\Delta_i$ in \eqref{eq:Deltaeq}. In addition, in view of Prop.~\ref{prop:globineq}, $\Delta_i(x;X) = 0$ for all $\mathrm{bd}(\cV^i) \backslash \mathrm{bd}(\cS)$, which implies that $\Delta_i^{\cI}(x) \geq 0$ for all $x\in\mathrm{bd}(\cV^i) \backslash \mathrm{bd}(\cS)$. Next, we show that the previous non-strict inequality can only hold as an equality. Let us assume that there exists $z \in \mathrm{bd}(\cV^i) \backslash \mathrm{bd}(\cS)$ such that $\Delta_i^{\cI}(z) > 0$. However, since $\Delta_i(z;X) = 0$, there is $j_z \notin \cI$ such that $\delta(z;x_i) = \delta(z;x_{j_z})$, which implies that the agent $j_z$ is a neighbor of the $i$-th agent, or equivalently, $j_z \in \cN_i$. However, $j_z \notin \cI$ and we know that, by hypothesis, $\cI \subseteq \cN_i$; thus, we have reached a contradiction and the proof is complete.
\end{proof}

\begin{remark}
Proposition \ref{prop:distrVP} implies that the Voronoi cell $\cV^i$ and its boundary $\mathrm{bd}(\cV^i)$, which are fully characterized in Proposition~\ref{prop:main1}, can be computed in a distributed way that relies on the exchange of information of the $i$-th agent with only the set of agents whose index belongs to $\widetilde{\cN}_i \supseteq \cN_i$ (the latter set of agents contains necessarily the set of neighbors of the $i$-th agent in view of Proposition~\ref{prop:neigh}). In other words, the cell $\cV^i$ and its boundary $\mathrm{bd}(\cV^i)$ can be computed in a \textit{distributed} way, which is a key result of this work.
\end{remark}

\begin{remark}
Let us assume that the $i$-th agent can communicate with all of its teammates in order to compute the point-set $P_{\theta}^i$, which according to Proposition~8 plays a key role in the complete characterization of $\cV^i$ and $\mathrm{bd}(\cV^i)$. For a given $\theta\in[0,2\pi[$, the point-set $P_{\theta}^i$ will consist of $M$ points, which means that the $i$-th agent will have to exchange at least $M$ messages with the other agents from the same network, under the assumption of an all-to-all type communication. To each pair $(i,j)$ corresponds at most two points in $P_{\theta}^i$ (note that the quadratic equation \eqref{eq:quadreq} has at most 2 solutions whose corresponding points $p = x_i+\rho e_{\theta}$ can lie in $\cS$). Thus, in the worst case, $M = 2 n$ assuming the exchange of 2 messages for each un-ordered pair $(i,j)$. The most expensive part of the proposed partitioning algorithm is the ordering of the points in $P_{\theta}^i$ in accordance with \eqref{p:orderline} to construct the (ordered) point-set $\mathfrak{P}_{\theta}^i$. The process of ordering the point-set $P_{\theta}^i$ (equivalent to sorting a list) has worst-case time complexity in $\mathcal{O}(M \mathrm{ln}(M))$ or $\mathcal{O}(2n \mathrm{ln}(2n))$. Let $n_i$ denote the number of the agents which are located in the set $\cA_i$, which according to Prop.~~\ref{prop:neincl} contains the locations of all the neighbors of the $i$-th agent. Now, let $\zeta_i = n_i/n$, then the number of messages that the $i$-th agent has to exchange is $\zeta_i M$ and the worst-time complexity for ordering the points of $P_{\theta}^i$ that lie in $\cA_i$ in accordance with \eqref{p:orderline} is in $\mathcal{O}(2\zeta_i n \, \mathrm{ln}(2\zeta_i n ))$.
\end{remark}

\section{Numerical Simulations}~\label{s:simu}
We consider a heterogeneous multi-agent network of $n=24$ agents
(plus the $0$-th agent) with different distance operators. For our
simulations, we consider the spatial domain $\cS = [-4,4]\times
[-4,4]$ and we take $\mathbf{P}_i = \mathbf{U}_i \mathbf{D}
\mathbf{U}_i\t$, with $\mathbf{D} = \big[
\begin{smallmatrix} 8 & 0\\ 0 & 3 \end{smallmatrix} \big]$
and $\mathbf{U}_i = \big[ \begin{smallmatrix} \cos \phi_i & -\sin \phi_i\\
\sin \phi_i & \cos \phi_i \end{smallmatrix} \big]$, where $\phi_i =
2\pi i/n$, for $i \in [1,n]_{\mZ}$, and $\mu_i=0$ for all $i \in
[1,n]_{\mZ}$. Clearly, $\lambda_{\min}(\mathbf{P}_i) = 3$ and
$\lambda_{\max}(\mathbf{P}_i) = 8$ for all $i\in[1,n]$ and thus, the
ratio $\lambda_{\max}(\mathbf{P}_i) / \lambda_{\min}(\mathbf{P}_i) =
8/3$, which indicates the presence of strong anisotropic features.
Furthermore, we take $x_0=(1/n) x_i$ (average position of the agents
of the actual network), $\mathbf{P}_0 = \lambda_0 \mathbf{I}$ with
$\lambda_0 \in \{1.7, 2.9\}$ and $\mu_0=0$ (note that $0< \lambda_0
< \lambda_{\min}(\mathbf{P}_i)$ for all $i \in [1, n]_{\mZ}$). With
this particular selection of parameters, both Assumptions~1 and 2
are clearly satisfied. The HQVPs generated by the positions of the
extended network are illustrated in Fig.~\ref{F:HQVPMAS1} for
$\lambda_0=1.7$ and in Fig.~\ref{F:HQVPMAS2} for $\lambda_0=2.9$.
The partitions in Figure~\ref{F:HQVPMAS} have been computed by means
of exhaustive numerical techniques and the obtained results are
included here mainly for verification purposes. In the same figure,
we have included contours (level sets) of the proximity metric of
each agent restricted on their own cells to illustrate the
anisotropic features in this partitioning problem. The cell $\cV^0$
corresponds to the red cell which is placed near the center of the
spatial domain $\cS$. We observe that $\cV^0$ is smaller when
$\lambda_0 =2.9$ than when $\lambda_0=1.7$. Note that by letting
$\lambda_0$ get closer (from below) to
$\lambda_{\min}(\mathbf{P}_i)=3$, the matrix $\mathbf{P}_0$ gets
``closer'' to violating Assumption~\ref{assumption2} whereas the
coverage hole $\cV^0$ becomes smaller. Thus, selection of
$\lambda_0$ has to strike a balance between well-posedness of the
proposed partitioning algorithm and smallness of the coverage hole
$\cV^0$. Another interesting observation is that the cell $\cV^{15}$ in both partitions is comprised of two disconnected components (only one of them contains in its interior the corresponding generator $x_{15}$).

Figure~\ref{F:HQVPDEC} illustrates the cells $\cV^{14}$ and
$\cV^{23}$ of the HQVP computed by means of the proposed distributed
algorithm for $\lambda_0 = 1.7$ (Figs.~\ref{F:HQVPDEC1}-\ref{F:HQVPDEC2}) and $\lambda_0 = 2.9$ (Figs.~\ref{F:HQVPDEC3}-\ref{F:HQVPDEC4}). For these simulations, we have used a uniform grid of $[0,2\pi[$ comprised of 360 nodes for the parameter (angle) $\theta$. The cross markers denote the
generators $x_{14}$ and $x_{23}$ whereas the small red circles and
red disks correspond to the positions of the rest of the agents of
the \textit{extended} network. In particular, the red (filled) disks in
Fig.~\ref{F:HQVPDEC} correspond to the neighbors
of the $i$-th agent in the topology of the HQVP, for $i=14$ and
$i=23$, respectively. The red dashed-dotted curves in the same figures indicate the boundaries of the ellipsoids $E_{14}$ and $E_{23}$
(recall that the latter ellipsoids contain the cells $\cV^{14}$ and
$\cV^{23}$ in view of Proposition~\ref{prop:baslemma}) whereas the
blue dashed curves denote the boundaries of the sets $\cA_{14}$ and
$\cA_{23}$ which contain the neighbors of the $i$-th agent for,
respectively, $i=14$ and $i=23$ in view of
Proposition~\ref{prop:neincl}. We observe that the cells $\cV^{14}$
and $\cV^{23}$ in Fig.~\ref{F:HQVPDEC} match
with their corresponding cells in Fig.~\ref{F:HQVPMAS1}. In
addition, the results illustrated in Fig.~\ref{F:HQVPDEC1}
--\ref{F:HQVPDEC4} are in agreement with
Propositions~\ref{prop:main} and~\ref{prop:neincl}. In particular,
the ellipsoids $E_{14}$ and $E_{23}$ contain, respectively, the
cells $\cV^{14}$ and $\cV^{23}$. Furthermore, the sets $\cA_{14}$
and $\cA_{23}$ contain the neighbors of the $i$-th agent for,
respectively, $i=14$ and $i=23$, which are denoted as filled red disks. 

We observe that 
the sets $E_{14}$, $E_{23}$, $\cA_{14}$
and $\cA_{23}$ in Figs.~\ref{F:HQVPDEC1}-\ref{F:HQVPDEC2} (corresponding to $\lambda_0=1.7$) are significantly smaller than their counterparts in Figs.~\ref{F:HQVPDEC3}-\ref{F:HQVPDEC4} (corresponding to $\lambda_0=2.9$). We conclude that although the decrease of the value of the parameter $\lambda_0$ may increase the size of the coverage hole (cell $\cV^0$), it may, on the other hand, render the problem of discovering the network topology induced by HQVP more meaningful in the sense that by solving the latter problem each agent will be able to identify a rather small subset of the spatial domain that necessarily contains its neighbors. In this way, each agent will be able to avoid communicating with non-neighboring agents which cannot contribute to the process of computing its own cells. In our simulations, we observe that while the cells $\cV^{14}$ for $\lambda_0=1.7$ and $\lambda_0=2.9$ are identical and their agents have the exact same sets of neighbors in both cases, the agent $i=14$ has to communicate with more agents (the ones that lie within the set $\cA_{14}$ in view of Prop.~\ref{prop:neincl}) and also search for the boundary points of its own cell over a larger set (in view of the Prop.~\ref{prop:main}, $\cV^{14}$ is a subset of $E_{14}$) when $\lambda_0=2.9$ than when $\lambda_0 = 1.7$. The situation is similar for $\cV^{23}$ although the changes on the sets $E_{23}$ and $\cA_{23}$ have a less substantial effect mainly because the agent $i=23$ is isolated from the majority of its teammates and is located close to the boundary of the spatial domain $\cS$.

\begin{figure}[htb]
\centering
\psfrag{x}[][][0.7]{$x$} \psfrag{y}[][][0.7]{$y$}
\psfrag{O}[][][0.7]{$~\cV^0$} \psfrag{V}[][][0.7]{$\cV^{15}~~~$}
\subfigure[$\mathbf{P}_0 = 1.7 \mathbf{I}$]%
{\epsfig{file =
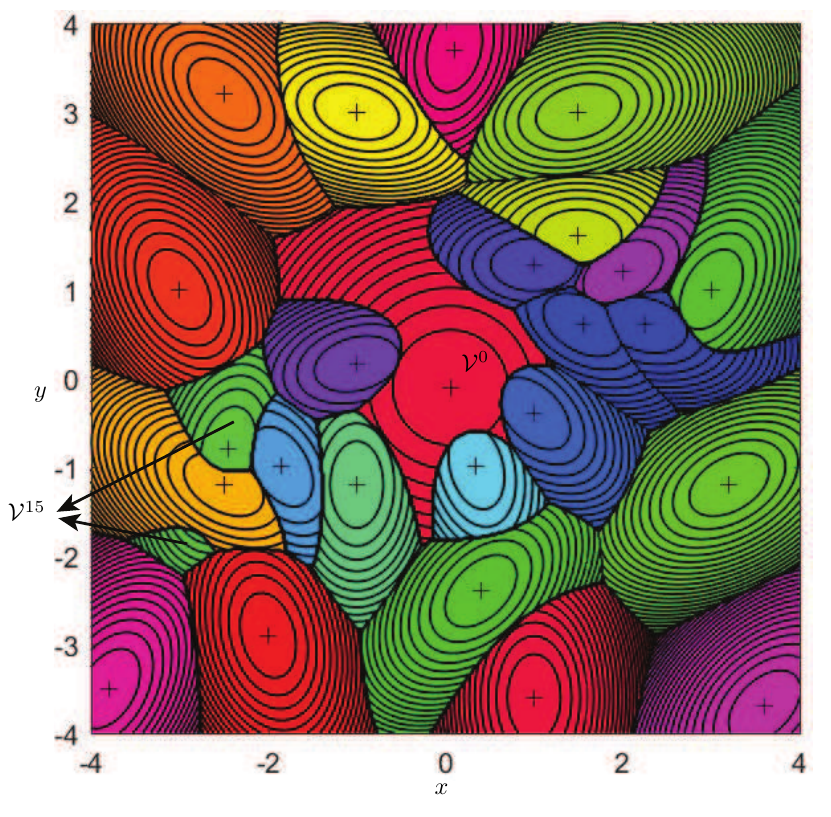,clip=,width=0.44\linewidth}\label{F:HQVPMAS1}}\hspace{3mm}
\subfigure[$\mathbf{P}_0 = 2.9 \mathbf{I}$]{\epsfig{file =
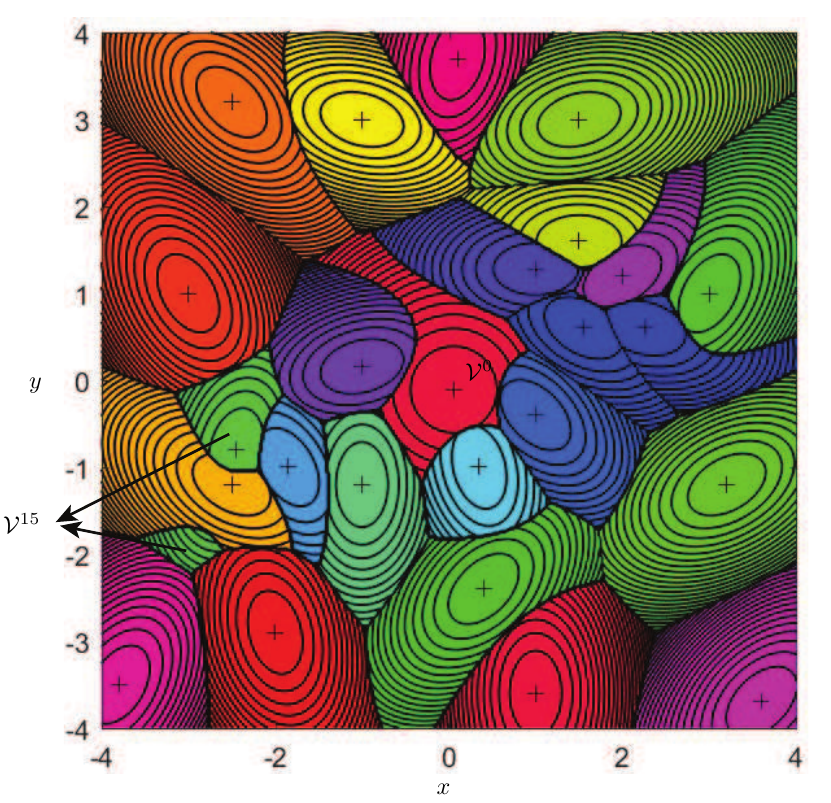,clip=,width=0.455\linewidth}\label{F:HQVPMAS2}}
\caption{\small{The HQVP generated by a heterogeneous network of
$n=24$ agents (plus the $0$-th agent).}}\label{F:HQVPMAS}
\end{figure}

\begin{figure}[htb]
\centering \psfrag{x}[][][0.8]{$x$} \psfrag{y}[][][0.8]{$y$}
\psfrag{a}[][][0.8]{$\cA_{14}$} \psfrag{A}[][][0.8]{$\cA_{23}$}
\psfrag{v}[][][0.8]{$\cV^{14}$} \psfrag{V}[][][0.8]{$~~~\cV^{23}$}
\psfrag{e}[][][0.8]{$E_{14}$} \psfrag{E}[][][0.8]{$E_{23}$}
\subfigure[$i = 14$ ($\lambda_0 =1.7$)]{\epsfig{file =
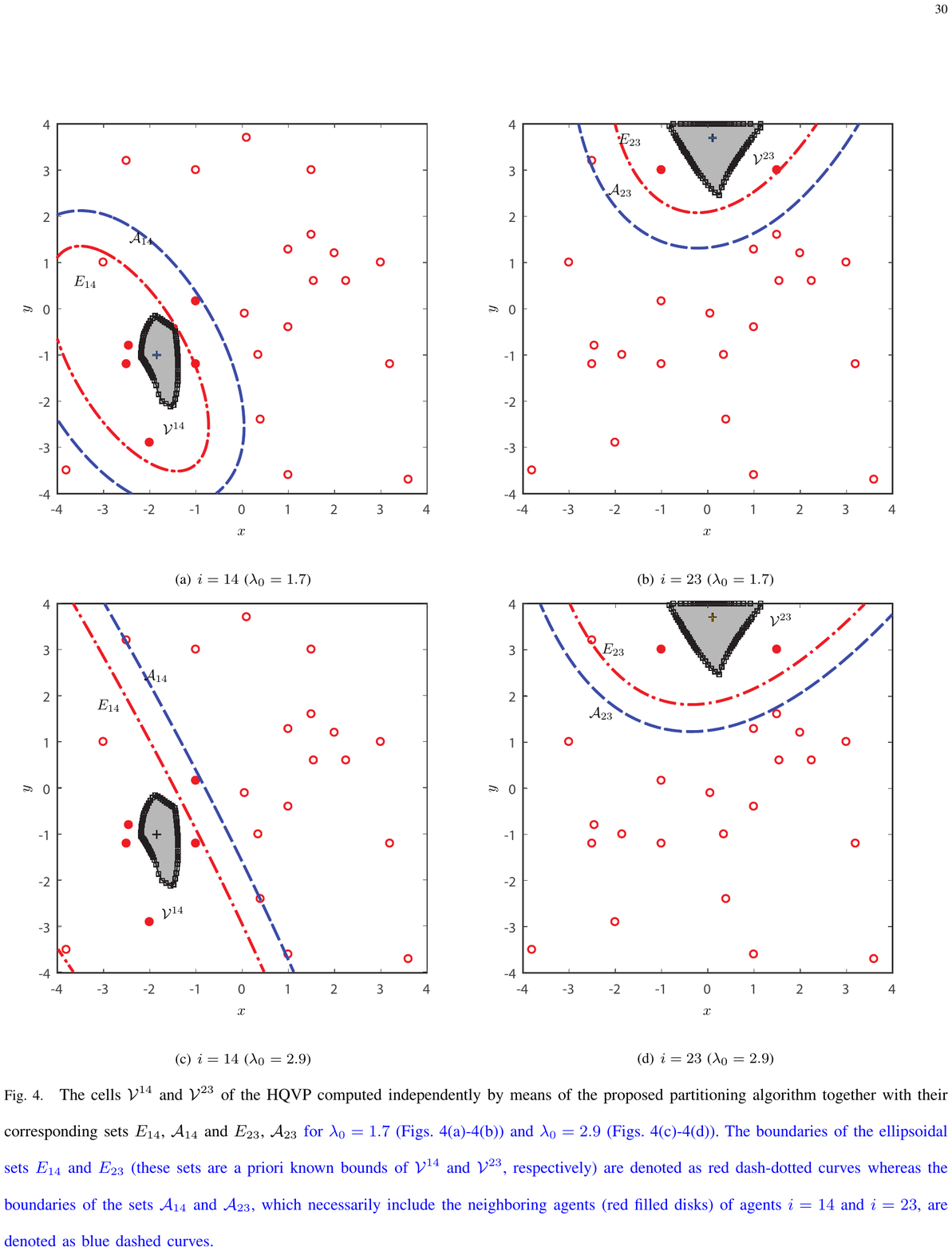,clip=,width=0.48\linewidth}\label{F:HQVPDEC1}}
\subfigure[$i = 23$ ($\lambda_0 =1.7$)]{ \epsfig{file =
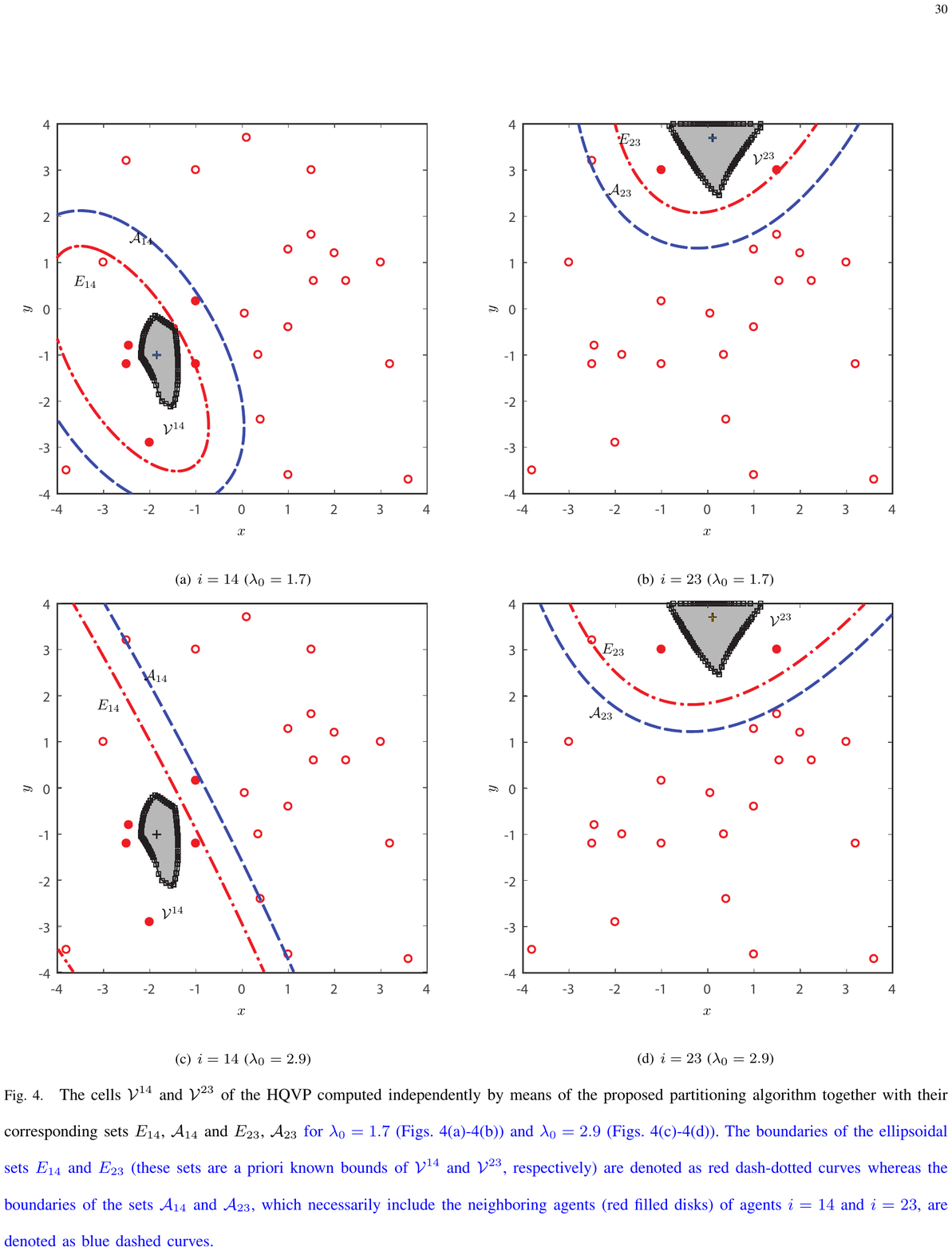,clip=,width=0.48\linewidth}\label{F:HQVPDEC2}}
\subfigure[$i = 14$ ($\lambda_0 =2.9$)]{\epsfig{file =
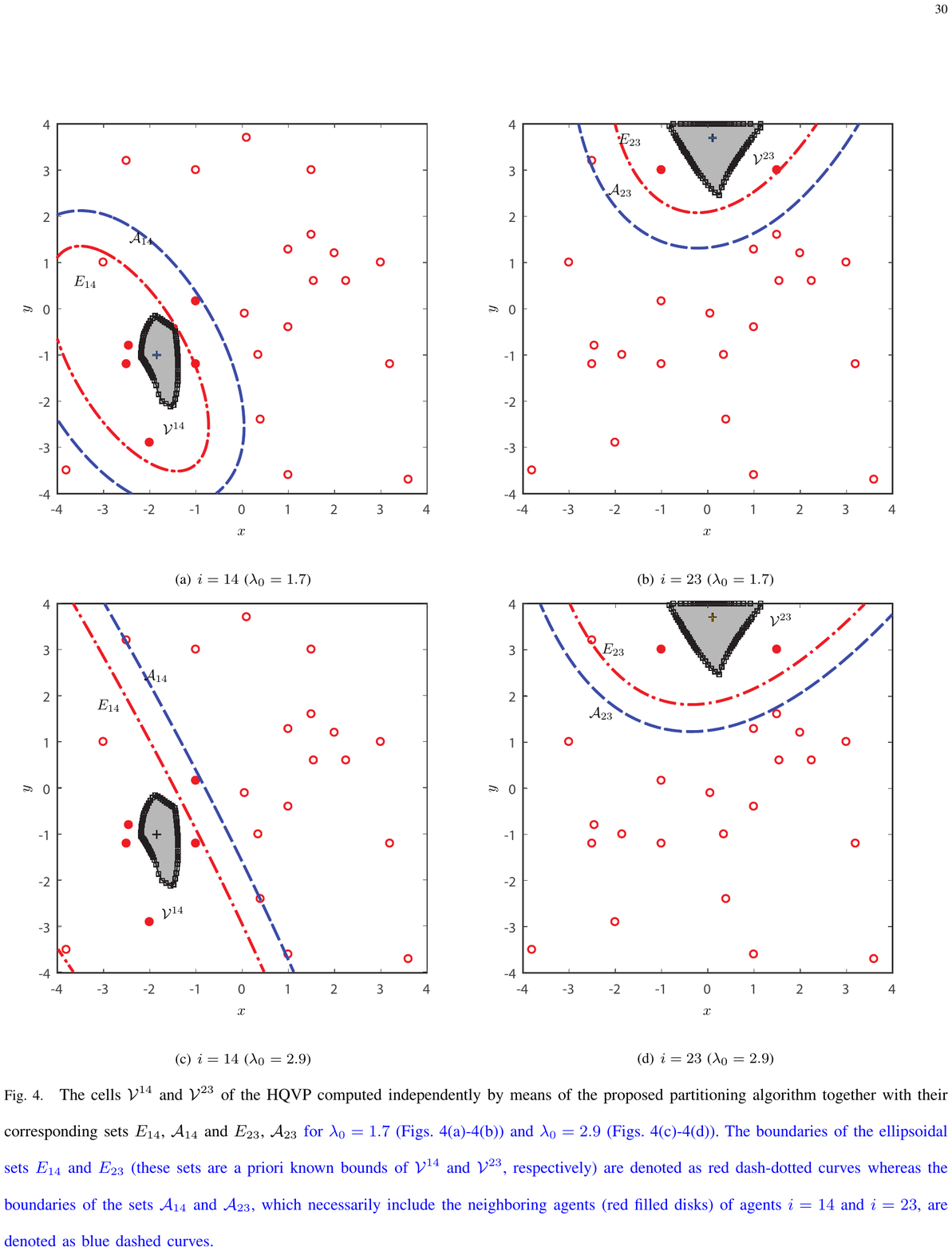,clip=,width=0.48\linewidth}\label{F:HQVPDEC3}}
\subfigure[$i = 23$ ($\lambda_0 =2.9$)]{ \epsfig{file =
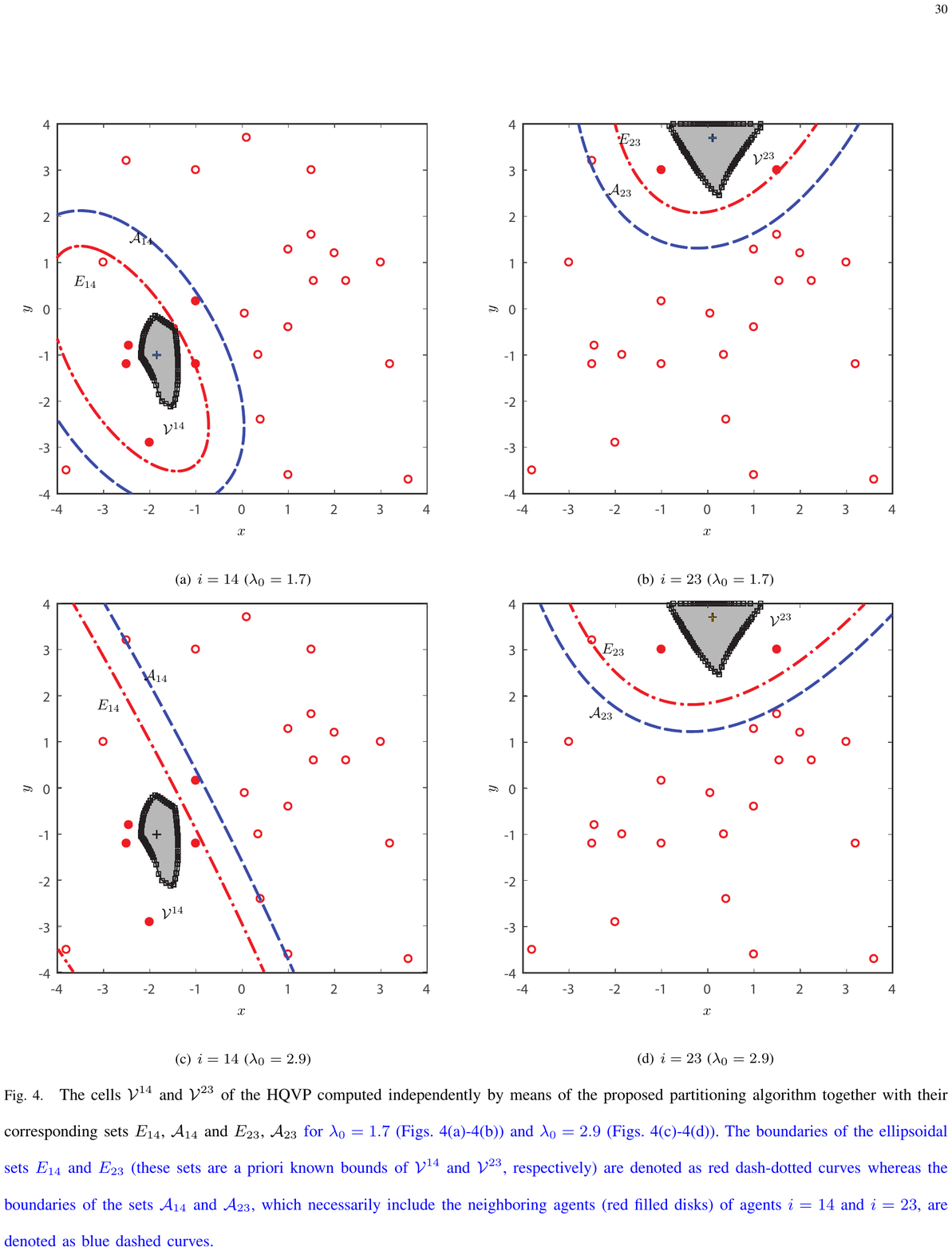,clip=,width=0.48\linewidth}\label{F:HQVPDEC4}}
\caption{\small{The cells $\cV^{14}$ and $\cV^{23}$ of the HQVP
computed independently by means of the proposed partitioning
algorithm together with their corresponding sets $E_{14}$, $\cA_{14}$ and $E_{23}$, $\cA_{23}$ for $\lambda_0 = 1.7$ (Figs.~\ref{F:HQVPDEC1}-\ref{F:HQVPDEC2}) and $\lambda_0 = 2.9$ (Figs.~\ref{F:HQVPDEC3}-\ref{F:HQVPDEC4}). The boundaries of the ellipsoidal sets $E_{14}$ and $E_{23}$ (these sets are a priori known bounds of $\cV^{14}$ and $\cV^{23}$, respectively) are denoted as red dash-dotted curves whereas the boundaries of the sets $\cA_{14}$ and $\cA_{23}$, which necessarily include the neighboring agents (red filled disks) of agents $i=14$ and $i=23$, are denoted as blue dashed curves.} }\label{F:HQVPDEC}
\end{figure}

\section{Conclusion}~\label{s:concl}
In this work, we have presented distributed algorithms for workspace
partitioning and network topology discovery problems for
heterogeneous multi-agent networks whose agents employ different
quadratic proximity metrics. The proposed algorithms leverage the
underlying structure of the solutions to the problems considered. In
our future work, we will explore how the proposed algorithms can be
integrated in solution techniques for distributed optimization and
estimation problems for heterogeneous networks operating in
anisotropic environments.

\bibliographystyle{ieeetr}
\bibliography{bakolas}

\begin{thebibliography}{10}

\bibitem{p:Labelle03}
F.~Labelle and J.~R. Shewchuk, ``Anisotropic \textsc{V}oronoi diagrams and
  guaranteed quality anisotropic mesh generation,'' in {\em SCG' 03},
  pp.~191--200, 2003.

\bibitem{p:Lloyd82}
S.~P. Lloyd, ``Least squares quantization in \textsc{PCM},'' {\em IEEE
  Transactions on Information Theory}, vol.~28, no.~2, pp.~129--137, 1982.

\bibitem{p:cortez}
J.~Cortes, S.~Martinez, T.~Karatas, and F.~Bullo, ``Coverage control for mobile
  sensing networks,'' {\em IEEE Transactions on Robotics and Automation},
  vol.~20, no.~2, pp.~243--255, 2004.

\bibitem{p:cortes05}
J.~Cortes, S.~Martinez, and F.~Bullo, ``Spatially-distributed coverage
  optimization and control with limited-range interactions,'' {\em ESAIM:
  COCV}, vol.~11, no.~4, pp.~691--719, 2005.

\bibitem{p:martinez06}
S.~Martinez and F.~Bullo, ``Optimal sensor placement and motion coordination
  for target tracking,'' {\em Automatica}, vol.~42, no.~4, pp.~661--668, 2006.

\bibitem{p:slotine2009}
M.~Schwager, D.~Rus, and J.-J. Slotine, ``Decentralized, adaptive coverage
  control for networked robots,'' {\em Int. J. Robot. Res.}, vol.~28, no.~3,
  pp.~357--375, 2009.

\bibitem{p:cortes2010}
J.~Cortes, ``Coverage optimization and spatial load balancing by robotic sensor
  networks,'' {\em IEEE Trans. Autom. Control}, vol.~55, no.~3, pp.~749--754,
  2010.

\bibitem{p:Breiten10}
A.~Breitenmoser, M.~Schwager, J.~C. Metzger, and D.~Rus, ``Distributed coverage
  and exploration in unknown non-convex environments,'' in {\em Proc. of the
  International Conference on Robotics and Automation}, (Anchorage, Alaska),
  pp.~4982--4989., May 2010.

\bibitem{p:pavoneTAC2011}
M.~Pavone, A.~Arsie, E.~Frazzoli, and F.~Bullo, ``Distributed algorithms for
  environment partitioning in mobile robotic networks,'' {\em IEEE Trans.
  Autom. Control}, vol.~56, no.~8, pp.~1834--1848, 2011.

\bibitem{p:Schwager2011}
M.~Schwager, D.~Rus, and J.-J. Slotine, ``Unifying geometric, probabilistic,
  and potential field approaches to multi-robot deployment,'' {\em Int. J.
  Robot. Res.}, vol.~30, no.~3, pp.~371--383, 2011.

\bibitem{fb-rc-pf:08u}
F.~Bullo, R.~Carli, and P.~Frasca, ``Gossip coverage control for robotic
  networks: {D}ynamical systems on the space of partitions,'' {\em SIAM Journal
  on Control and Optimization}, vol.~50, no.~1, pp.~419--447, 2012.

\bibitem{rp-pf-rb:13i}
R.~Patel, P.~Frasca, and F.~Bullo, ``Centroidal area-constrained partitioning
  for robotic networks,'' {\em ASME Journal of Dynamic Systems, Measurement,
  and Control}, vol.~136, no.~3, p.~031024, 2014.

\bibitem{p:tzes13}
Y.~Stergiopoulos and A.~Tzes, ``Spatially distributed area coverage
  optimisation in mobile robotic networks with arbitrary convex anisotropic
  patterns,'' {\em Automatica}, vol.~49, no.~1, pp.~232--237, 2013.

\bibitem{p:bhatta2014}
S.~Bhattacharya, R.~Ghrist, and V.~Vijay~Kumar, ``Multi-robot coverage and
  exploration on \textsc{R}iemannian manifolds with boundaries,'' {\em Int. J.
  Robot. Res.}, vol.~33, no.~1, pp.~113--137, 2014.

\bibitem{p:tvlloyd15}
S.~G. Lee, Y.~Diaz-Mercado, and M.~Egerstedt, ``Multirobot control using
  time-varying density functions,'' {\em IEEE Transactions on Robotics},
  vol.~31, pp.~489--493, April 2015.

\bibitem{p:anisosens2008}
A.~Gusrialdi, S.~Hirche, T.~Hatanaka, and M.~Fujita, ``Voronoi based coverage
  control with anisotropic sensors,'' in {\em American Control Conference},
  pp.~736--741, June 2008.

\bibitem{p:hadji}
M.~Cao and C.~N. Hadjicostis, ``Distributed algorithms for \textsc{V}oronoi
  diagrams and applications in ad-hoc networks,'' Technical Report
  UILU-ENG-03-22222160, UIUC Coordinated Science Laboratory, 2003.

\bibitem{p:lynchvoro2017}
M.~L. {Elwin}, R.~A. {Freeman}, and K.~M. {Lynch}, ``Distributed voronoi
  neighbor identification from inter-robot distances,'' {\em IEEE Robotics and
  Automation Letters}, vol.~2, pp.~1320--1327, July 2017.

\bibitem{p:bt_autom10}
E.~Bakolas and P.~Tsiotras, ``The \textsc{Z}ermelo-\textsc{V}oronoi diagram: a
  dynamic partition problem,'' {\em Automatica}, vol.~46, no.~12,
  pp.~2059--2067, 2010.

\bibitem{p:baktsiaut13}
E.~Bakolas and P.~Tsiotras, ``Optimal partitioning for spatiotemporal coverage
  in a drift field,'' {\em Automatica}, vol.~49, no.~7, pp.~2064--2073, 2013.

\bibitem{p:bak2018}
E.~Bakolas, ``Distributed partitioning algorithms for locational optimization
  of multiagent networks in \textsc{SE}(2),'' {\em IEEE Transactions on
  Automatic Control}, vol.~63, no.~1, pp.~101--116, 2018.

\bibitem{p:bakolas2013b}
E.~Bakolas, ``Optimal partitioning for multi-vehicle systems using quadratic
  performance criteria,'' {\em Automatica}, vol.~49, no.~11, pp.~3377--3383,
  2013.

\bibitem{p:BAKOLAS2014}
E.~Bakolas, ``Decentralized spatial partitioning algorithms for multi-vehicle
  systems based on the minimum control effort metric,'' {\em Systems \& Control
  Letters}, vol.~73, pp.~81--87, 2014.

\bibitem{p:BAKOLAS2016}
E.~Bakolas, ``Distributed partitioning algorithms for multi-agent networks with
  quadratic proximity metrics and sensing constraints,'' {\em Systems \&
  Control Letters}, vol.~91, pp.~36--42, 2016.

\bibitem{p:bakauto2014}
E.~Bakolas, ``Decentralized spatial partitioning for multi-vehicle systems in
  spatiotemporal flow-field,'' {\em Automatica}, vol.~50, no.~9,
  pp.~2389--2396, 2014.

\bibitem{p:bakaut2015}
E.~Bakolas, ``Partitioning algorithms for multi-agent systems based on
  finite-time proximity metrics,'' {\em Automatica}, vol.~55, pp.~176--182,
  2015.

\bibitem{p:Hoff:1999}
K.~E. Hoff, III, J.~Keyser, M.~Lin, D.~Manocha, and T.~Culver, ``Fast
  computation of generalized \textsc{V}oronoi diagrams using graphics
  hardware,'' in {\em SIGGRAPH '99}, (New York, NY, USA), pp.~277--286, 1999.

\bibitem{p:arslan2019}
O.~Arslan, ``Statistical coverage control of mobile sensor networks,'' {\em
  IEEE Transactions on Robotics}, pp.~1--20, 2019.

\bibitem{p:voronoi}
G.~F. Voronoi, ``Nouveles applications des param\`{e}tres continus \`{a} la
  th\'{e}orie de formas quadratiques,'' {\em Journal f\"{u}r die Reine und
  Angewandte Mathematik}, vol.~134, pp.~198--287, 1908.

\bibitem{b:boigeom}
J.-D. Boissonnat and M.~Yvinec, {\em Algorithmic Geometry}.
\newblock Cambridge, United Kingdom: Cambridge University Press, 1998.

\bibitem{p:BoissoAnisoVor2008}
J.-D. Boissonnat, C.~Wormser, and M.~Yvinec, ``Anisotropic diagrams:
  \textsc{L}abelle \textsc{S}hewchuk approach revisited,'' {\em Theoretical
  Computer Science}, vol.~408, no.~2--3, pp.~163--173, 2008.

\bibitem{p:consepropa2006}
C.~C. {Moallemi} and B.~{Van Roy}, ``Consensus propagation,'' {\em IEEE Trans.
  Inf. Theory}, vol.~52, no.~11, pp.~4753--4766, 2006.

\bibitem{p:XIAO2007}
L.~Xiao, S.~Boyd, and S.-J. Kim, ``Distributed average consensus with
  least-mean-square deviation,'' {\em Journal of Parallel and Distributed
  Computing}, vol.~67, no.~1, pp.~33--46, 2007.

\bibitem{b:lynch1996}
N.~A. Lynch, {\em Distributed algorithms}.
\newblock Elsevier, 1996.

\bibitem{p:colavoidell17}
B.~H. Lee, J.~D. Jeon, and J.~H. Oh, ``Velocity obstacle based local collision
  avoidance for a holonomic elliptic robot,'' {\em Autonomous Robots}, vol.~41,
  no.~6, pp.~1347--1363, 2017.

\end{thebibliography}

\end{document}